\documentclass[aps,amsfonts,amssymb,10pt,letterpaper, 
final, nofootinbib, longbibliography, accepted=2020-03-24,
twocolumn,notitlepage,hyperref,twoside,allowtoday,noarxiv]{quantumarticle}

\usepackage[showerrors,immediate]{silence}
\usepackage[utf8]{inputenx}
\usepackage[OT1]{fontenc}
\usepackage[english]{babel}
\usepackage[style=american,autopunct=true]{csquotes}
\usepackage{soul}

\usepackage{bold-extra}
\usepackage{appendix}

\usepackage{amsmath, amssymb, amsthm,verbatim, graphicx,bbm}
\DeclareMathAlphabet\mathbfcal{OMS}{cmsy}{b}{n}
\usepackage{epsfig}
\usepackage{textcomp}
\usepackage{mathrsfs, mathtools}
\usepackage{units}
\usepackage{fullpage, stackrel, subfigure}
\usepackage{paralist}
\usepackage{comment}
\usepackage{placeins}

\usepackage{subfigure}
\usepackage{mathrsfs}
\usepackage[all]{xy}

\newcommand{\too}{\! \! \to \! \!}

\usepackage{tikz}

\usetikzlibrary{calc,decorations.pathreplacing, intersections,through, patterns}
\usetikzlibrary{arrows,shapes}
\usetikzlibrary{calc,decorations.pathreplacing}
\definecolor{darkgreen}{rgb}{0,.5,0}

\usetikzlibrary{backgrounds}
\pgfdeclarelayer{myback}
\pgfsetlayers{background,myback,main} 

\usepackage{hyperref}

\usepackage{tabularx}
\usepackage{booktabs} 
\newcolumntype{R}{>{\raggedleft\arraybackslash}X}
\newcolumntype{C}{>{\centering\arraybackslash}X}
\newcolumntype{L}{>{\raggedright\arraybackslash}X}

\usepackage{tikz}
\usetikzlibrary{calc,decorations.pathreplacing,decorations.markings}
\usetikzlibrary{arrows,shapes}
\usepackage{accents}

\pgfkeys{/tikz/.cd,
  circle color/.initial=black,
  circle color/.get=\circlecolor,
  circle color/.store in=\circlecolor,
}

\tikzset{dotted pattern/.style args={#1 and #2}{
   postaction=decorate,
   decoration={
    markings,
    mark=
    between positions 0 and 1 step #2
      with
      {
       \fill[radius=#1,\circlecolor] (0,0) circle;
      }
    }
  },
  dotted pattern/.default={1pt and 1.5mm},
}

\newcommand{\beq}{\begin{equation}}
\newcommand{\eeq}{\end{equation}}

\newcommand{\LOSR}[0]{\ifmmode\textup{\upshape LOSR}\else{\textup{\upshape LOSR}}\fi}
\newcommand{\DetLOSR}[0]{\ifmmode\textup{\upshape LDO}\else{\textup{\upshape LDO}}\fi}
\newcommand{\LDO}[0]{\ifmmode\textup{\upshape LDO}\else{\textup{\upshape LDO}}\fi}
\newcommand{\LSO}[0]{\ifmmode\textup{\upshape LSO}\else{\textup{\upshape LSO}}\fi}
\newcommand{\LDTNO}[0]{\ifmmode\textup{\upshape LDTNO}\else{\textup{\upshape LDTNO}}\fi}
\newcommand{\LO}[0]{\ifmmode\textup{\upshape LO}\else{\textup{\upshape LO}}\fi}
\newcommand{\LOCC}[0]{\ifmmode\textup{\upshape LOCC}\else{\textup{\upshape LOCC}}\fi}

\newcommand{\nhphantom}[1]{\sbox0{#1}\hspace{-\the\wd0}}

\newcommand*{\LOSRconv}{\xmapsto{\LOSR}}

\usepackage[only,Arrownot]{stmaryrd}

\newcommand*{\LOSRinterconv}{\xleftrightarrow{\LOSR}}

\usepackage{braket}

\newtheorem{theo}{Theorem}
\newtheorem{thm}[theo]{Theorem}

\newtheorem{cor}{Corollary}

\newtheorem{defn}{Definition}

\theoremstyle{definition} 

\theoremstyle{plain}
\newtheorem{innercustomgeneric}{\customgenericname}
\providecommand{\customgenericname}{}
\newcommand{\newcustomtheorem}[2]{%
  \newenvironment{#1}[1]
  {%
   \renewcommand\customgenericname{#2}%
   \renewcommand\theinnercustomgeneric{##1}%
   \innercustomgeneric
  }
  {\endinnercustomgeneric}
}

\newcustomtheorem{customprop}{Proposition}
\newcustomtheorem{customcor}{Corollary}
\newcustomtheorem{customlem}{Lemma}
\newcustomtheorem{customopq}{Open Question}

\begin{document}
\title{The type-independent resource theory of local operations and shared randomness}
\author{David Schmid}
\affiliation{Perimeter Institute for Theoretical Physics, 31 Caroline St. N, Waterloo, Ontario, N2L 2Y5, Canada}
\affiliation{Institute for Quantum Computing and Dept. of Physics and Astronomy, University of Waterloo, Waterloo, Ontario N2L 3G1, Canada}
\email{dschmid@perimeterinstitute.ca}
\author{Denis Rosset}
\affiliation{Perimeter Institute for Theoretical Physics, 31 Caroline St. N, Waterloo, Ontario, N2L 2Y5, Canada}
\author{Francesco Buscemi}
\affiliation{Graduate School of Informatics, Nagoya University, Chikusa-ku, 464-8601 Nagoya, Japan}
\begin{abstract}
In space-like separated experiments and other scenarios where multiple parties share a classical common cause but no cause-effect relations, quantum theory allows a variety of nonsignaling resources which are useful for distributed quantum information processing. These include quantum states, nonlocal boxes, steering assemblages, teleportages, channel steering assemblages, and so on. Such resources are often studied using nonlocal games, semiquantum games, entanglement-witnesses, teleportation experiments, and similar tasks.
We introduce a unifying framework which subsumes the full range of nonsignaling resources, as well as the games and experiments which probe them, into a common resource theory: that of local operations and shared randomness (LOSR). Crucially, we allow these LOSR operations to locally change the type of a resource, so that players can convert resources of {\em any} type into resources of any other type, and in particular into strategies for the specific type of game they are playing.
We then prove several theorems relating resources and games of different types. These theorems generalize a number of seminal results from the literature, and can be applied to lessen the assumptions needed to characterize the nonclassicality of  resources. As just one example, we prove that semiquantum games are able to perfectly characterize the LOSR nonclassicality of every resource of {\em any} type (not just quantum states, as was previously shown). As a consequence, we show that any resource can be characterized in a measurement-device-independent manner. \end{abstract}
\maketitle


\section{Introduction}
A key focus in quantum foundations is the study of nonclassicality.
Starting from the Einstein-Podolsky-Rosen paradox~\cite{Einstein1935}, special focus has been given to experiments involving space-like separated subsystems.
In the modern language of causality~\cite{Wood2015,Costa_2016,Allen2017,Lorenz2019}, the key feature of these scenarios is that the subsystems which are being probed share a classical common cause, but do not share any cause-effect channels between them.
In such scenarios, quantum theory allows for distributed quantum channels which act as valuable nonclassical resources for accomplishing tasks which would otherwise be impossible. 

The most common examples of such resources are entangled quantum states~\cite{horodecki2009quantum} and boxes producing nonlocal correlations~\cite{brunner2013Bell}; but there are many other types of useful resources.
We develop a resource-theoretic~\cite{resthry} framework which unifies a wide variety of these, including quantum states~\cite{Watrous}, boxes~\cite{brunner2013Bell}, steering assemblages~\cite{wisesteer,Cavalcanti2017}, channel steering assemblages~\cite{channelsteer}, teleportages~\cite{Hoban_2018,telep}, distributed measurements~\cite{Bennett}, measurement-device-independent steering channels~\cite{steer}, Bob-with-input steering channels~\cite{BobWI}, and generic no-signaling quantum channels~\cite{Watrous}.
Free (or classical) resources are those that can be generated freely by local operations and shared randomness (LOSR), encompassing the specific cases of separable quantum states, local boxes, unsteerable assemblages, and so on. Any resource which cannot be simulated by LOSR operations is said to be nonfree, or nonclassical. A resource is said to be at least as nonclassical as another resource if it can be transformed to the second using LOSR transformations. Crucially, such comparisons can be made for resources of arbitrary and potentially differing types.

Some works in the past have focused on LOSR as a resource theory {\em in specific scenarios}, such as for quantum states~\cite{sq,LOSRvsLOCCentang}, for nonlocal correlations~\cite{de2014nonlocality,gallego2016nonlocality,Bellquantified,LOSRvsLOCCentang}, and for steering assemblages~\cite{steer} (albeit under a different name). These previous works focused on one or two types of resources, and most commonly on quantum states. Our framework is more general, but subsumes each of these as a special case. 

In addition to introducing this encompassing framework, our second primary goal herein is to study how the type of a resource impacts the methods by which one can characterize its nonclassicality in practice.
For example, nonlocal boxes have classical inputs and outputs, and so only weak assumptions~\cite{Bell2004,Putz2014} about one's laboratory instruments are required for their characterization. 
However, when a resource has a quantum output, one requires a well-characterized quantum measurement to probe that output and consequently the resource~\cite{Rosset2012a}.
In such a case, the test of nonclassicality is said to be {\em device-dependent}, while in adversarial scenarios such as cryptography, the terminology of {\em trust} is also used~\cite{Pironio2016}.
The same idea applies to a quantum input, which must be probed using a well-characterized quantum state preparation device.
Thus, only nonlocal boxes can be probed in a {\em device-independent} manner; {\em a priori}, quantum states require well-characterized quantum measurement devices; while other objects, such as steering assemblages, require a mixture of both~\cite{Cavalcanti2015}.
Consequently, it is important to determine under what circumstances devices of one type may be converted into devices of a second type {\em in a manner that does not degrade their usefulness as a resource}.
If such a conversion is possible,
then one may be able to lessen the assumptions and technological requirements needed to characterize one's devices.

In some particular cases, previous work has studied this question of whether the nonclassicality of a quantum state can be characterized by first applying free operations which convert it to another type of resource.
For example, we know that some Werner states~\cite{Werner, Barrettlocal} have a local model for all measurements; such nonclassical states can only be transformed into classical boxes, and so all information about their nonclassicality is lost in the conversion. 
In contrast, the main result of Ref.~\cite{sq} proves that every entangled state can have its nonclassicality encoded in a semiquantum channel. Additionally, in Ref.~\cite{telep}, it is shown that every entangled state can generate a type of no-signaling channel (recently termed a teleportage~\cite{Hoban_2018}) which could not be generated by any separable state and which is useful for some task related to quantum teleportation~\cite{LipkaBartosik2019}.

It is useful to distinguish between qualitative versus quantitative characterizations of nonclassicality. 
To highlight the distinction, it is instructive to examine one particular line of research.
Ref.~\cite{sq} is often advertised as proving that the nonclassicality of every entangled state can be revealed in a generalization of nonlocal games termed {\em semiquantum games} (which were later used to construct {\em measurement-device-independent entanglement witnesses}~\cite{Branciard2013}).
However, this claim is actually a (qualitative) corollary of the (quantitative) main theorem, which showed that the performance of states in semiquantum games {\em exactly} reproduces the classification of entangled states under LOSR transformations.
Subsequent works~\cite{Branciard2013,Rosset2013a} focused on the qualitative distinction between classical and nonclassical resources, but still later works reinterpreted the payoffs of semiquantum games as measures of entanglement~\cite{Shahandeh2017,Rosset2018a}, thus reconnecting with the quantitative nature of Buscemi's original work.
Note also that the quantitative study of entanglement is historically linked to entanglement monotones~\cite{vidal2000entanglement}.
However, the study of nonclassicality cannot be reduced to a single such measure, as there are many inequivalent species of nonclassicality even in the simplest cases~\cite{Bellquantified}. Informed by the recent formalization of resource theories~\cite{resthry}, we study the fundamental mathematical object---the preorder of resources under LOSR transformations. One can then derive specific nonclassicality witnesses and monotones~\cite{rosset2019characterizing}, each of which provides an incomplete characterization of the preorder.

As implied just above, the mathematical structure which best allows for comparison between objects that need not be strictly ordered is a {\em preorder}.
Formally, a preorder is an ordering relation that is reflexive ($a \succeq a$) and transitive ($a \succeq b$ and $b \succeq c$ implies $a \succeq c$)\footnote{A preorder is distinguished from a partial order by the fact that $a \succeq b$ and $b \succeq a$ need not imply $a = b$. In a partial order, $a \succeq b$ and $b \succeq a$ implies $a = b$.}.
Our work focuses on three distinct preorders, which the reader should be careful to distinguish.
First, there is the preorder $R \succeq_{\rm LOSR} R'$ (sometimes denoted $R \LOSRconv R'$) that indicates if a resource $R$ can be converted into another resource $R'$ by LOSR transformations (Definition~\ref{LOSRorder}).
Second, there is the preorder $\succeq_{\rm type}$ over resource types that orders those types according to their ability to encode nonclassicality (Definition~\ref{typeorder}).
Finally, there is the preorder $\succeq_{\mathcal{G}_{\! T}}$ that ranks resources according to their performance with respect to the set $\mathcal{G}_{\! T}$ of all games of a particular type $T$ (Definition~\ref{orderwrtgame}).

This paper is best read alongside Ref.~\cite{rosset2019characterizing}.
In the current paper, we present a general framework to study quantum resources of arbitrary types, and we quantify the nonclassicality of these resources within a type-independent resource theory of local operations and shared randomness. Here, our main results center on showing how resources of one type can be more easily characterized by first converting them to resources of a second type. 
In Ref.~\cite{rosset2019characterizing}, our aim is practical and computational, focusing on how data can be used to characterize one's resources using off-the-shelf software. There, we include type-independent techniques for computing witnesses which can certify the nonclassicality of a resource, as well as techniques for computing the value of type-independent monotones (which we introduce therein).

\subsection{Organization of the paper}


In Section~\ref{sectypes}, we discuss various types of resources. We inventory the 9 possible types of a single party's partition of a resource, where that party's input and output may each be trivial, classical, or quantum. 
Focusing on the 81 bipartite resource types for simplicity, we recognize 10 types that have been studied in the literature and identify 5 new nontrivial resource types. All other bipartite resource types are either trivial or equivalent up to a symmetry.
We then define LOSR transformations between resources of arbitrary types, as well as the ordering over resources that this induces.

In Section~\ref{expressivity}, we define a precise sense in which some types can express the LOSR nonclassicality  of other types.
In many cases, conversions from a resource of one type to another type necessarily degrade the nonclassicality of the resource, as in Werner's example.
In other cases, one can perfectly encode the nonclassicality of any given resource into some resource of the target type, as in Buscemi's example. 
For every single-party type, we ask which can perfectly encode the nonclassicality of which others, and we answer this question for almost every pair, with the exception of one open question. From these considerations of single party types, one can deduce encodings of more complicated resource types which involve multiple parties.
Most strikingly, we show that semiquantum channels (with quantum inputs and classical outputs) are universal, in the sense that the nonclassicality of all resources can be encoded into them. 

In Section~\ref{unifiedgames}, we give an abstract framework for probing the nonclassicality of resources, subsuming as special cases the notions of nonlocal games~\cite{Bellreview}, semiquantum games~\cite{sq}, steering~\cite{wisesteer,steer} and teleportation~\cite{telep} experiments, and entanglement witnessing~\cite{Chruscinski2014}.
In our framework, every type of resource has a corresponding type of game, where a game of some type maps every resource of that type to a real number. (E.g., in nonlocal and semiquantum games, this number is the usual average game payoff).
We then show how resources of any type can be used to play a game designed for one specific type.
In some cases, games of one type can {\em completely} characterize the nonclassicality of every resource of another type. 
For example, Ref.~\cite{sq} showed that the LOSR nonclassicality of every quantum state is perfectly characterized by the set of semiquantum games.
We generalize these ideas by proving that if one type can encode another, then games of the first type can perfectly characterize the LOSR nonclassicality of all resources of the second type.
Together with our results on which types can encode which others, this expands the known methods for quantifying LOSR nonclassicality in practice and in theory. For example, our result on the universality of the semiquantum type implies that any resource of any type can be characterized by some semiquantum game, and hence can be characterized in a measurement-device-independent manner.

In Section~\ref{extending}, we relate our work to existing results.
First, we note how our results generalize the main result of Ref.~\cite{sq}, showing that semiquantum games can completely characterize the LOSR nonclassicality of arbitrary resource, not just of quantum states. 
Next, we show that the results of Ref.~\cite{steer} are a special case of two of our theorems when one applies steering experiments to quantify the nonclassicality of quantum states; further, our theorems provide a generalization of these arguments to more general experiments and types of resources.
Finally, we show that the LOSR nonclassicality of every quantum state is {\em completely} characterized by the set of teleportation games, and thus that the results of Ref.~\cite{telep} can be extended to be quantitative as well as qualitative. 

\section{Resource types and LOSR transformations between them} \label{sectypes}

We are interested in scenarios where the relevant parties share a classical common cause but do not share any cause-effect channels. For example, parties who perform experiments at space-like separation cannot access classical communication. 
For simplicity, we henceforth focus on bipartite scenarios; however, all of our results generalize immediately to arbitrarily many parties. We will consider only nonsignaling resources~\cite{Popescu1994,Barrett2005} throughout this work.\footnote{In fact, if one wishes to interpret resourcefulness as {\em nonclassicality}, then one must further restrict the enveloping theory to those resources which can be generated by local operations and quantum common causes. 
For non-signaling resources that {\em cannot} be realized in this manner~\cite{causallocaliz}, resourcefulness may originate in the nonclassicality of a common-cause process {\em or} in {\em classical} communication channels (which are fine-tuned so as to not exhibit signaling).}
We will not specifically consider post-quantum channels in this work, although one might naturally extend our work to include these as resources.
Hence, in this work a resource is a completely positive~\cite{NielsenAndChuang,Schmidcausal}, trace-preserving, nonsignaling quantum channel.
The parties may share various types of resources, which we now classify by type.

\subsection{Partition-types and global types}
In this paper, we use the term {\bf type} (of a resource) to refer exclusively to whether the various input and output systems are trivial ($\mathsf{I}$), classical ($\mathsf{C}$), or quantum ($\mathsf{Q}$). A system is said to be trivial if it has dimension one, is said to be classical if all operators on its Hilbert space are diagonal, and is otherwise said to be quantum.
(See Ref.~\cite{rosset2019characterizing} for more details.)
Additionally, if a resource has more than one input (output), which may be of different types, we imagine grouping them together, yielding an effective input (output) whose type is the least expressive type which embeds all those in the grouping, where quantum systems embed classical systems, which embed trivial systems.

We will denote the type of a single party's share of a resource by $T_i := X_i \too Y_i$, where $i$ labels the party and $X,Y \in \{\mathsf{I}, \mathsf{C}, \mathsf{Q} \}$, with $X$ labeling whether the input to that party is trivial ($\mathsf{I}$), classical ($\mathsf{C}$), or quantum ($\mathsf{Q}$) and $Y$ labeling the output similarly. We will refer to $T_i$ as the {\bf partition-type} of party $i$. 

We can then denote the {\bf global type} of an $n$-party resource as $T := T_1 T_2 ...T_n \simeq X_1X_2 ... X_n\too Y_1 Y_2...Y_n $. 
Note that while the specification of the global type of a resource fixes the number of parties and the types of their partitions of the resource, the specification of a partition-type does not constrain either the number of other parties who share the resource, nor the types of those other partitions. One could also consider partition-types for partitions of a resource which involve more than one party, but this paper makes use only of partition-types which involve a single party.

We now describe the ten examples of resource types from Fig.~\ref{types}, setting up some explicit terminology and conventions as we go. 
We graphically depict trivial, classical, and quantum systems by the lack of a wire, a single wire, and a double wire, respectively.

\begin{figure}[htb!]
\centering
\includegraphics[width=0.48\textwidth]{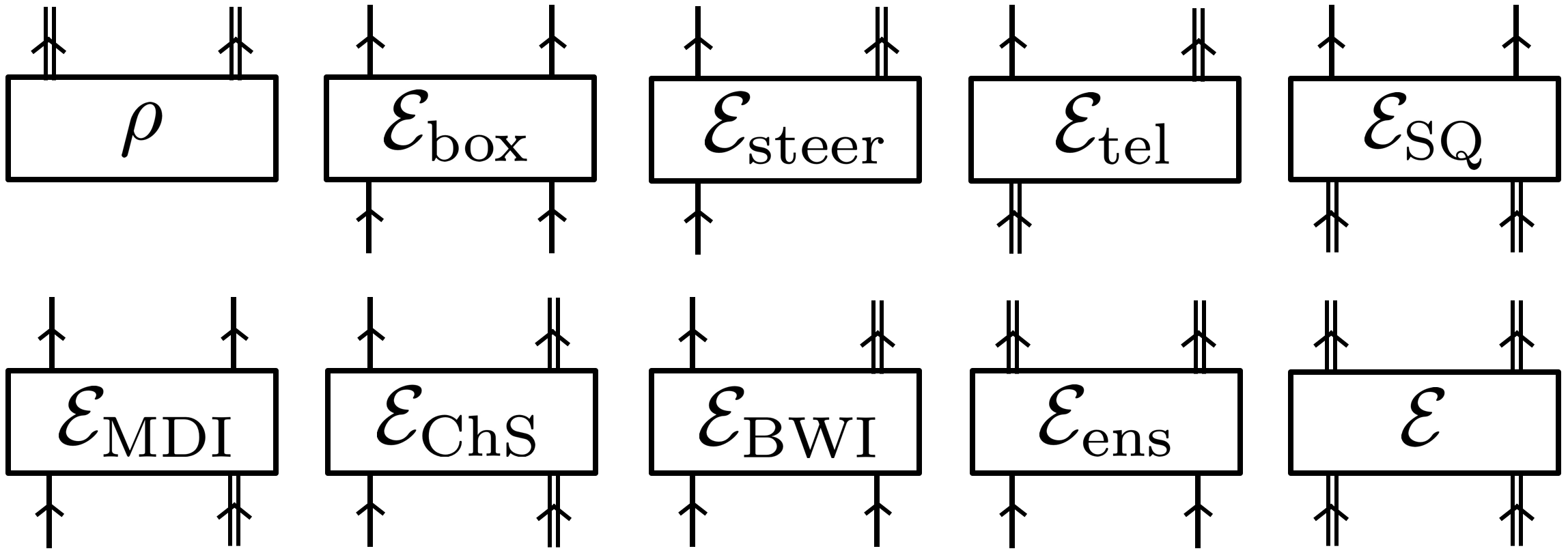}
\caption{Common types of no-signaling resources, where classical systems are represented by single wires and quantum systems are represented by double wires. (a) A quantum state $\rho$ has type $\mathsf{II} \too \mathsf{QQ}$. (b) A box $\mathcal{E}_{\rm box}$ has type $\mathsf{CC} \too \mathsf{CC}$. (c) A steering assemblage $\mathcal{E}_{\rm steer}$ has type $\mathsf{CI} \too \mathsf{CQ}$. (d) A teleportage $\mathcal{E}_{\rm tel}$ has type $\mathsf{QI} \too \mathsf{CQ}$. (e) A  semiquantum channel $\mathcal{E}_{\rm SQ}$ has type $\mathsf{QQ} \too \mathsf{CC}$. (f) A measurement-device-independent steering channel $\mathcal{E}_{\rm MDI}$ has type $\mathsf{CQ} \too \mathsf{CC}$. (g) A channel steering assemblage $\mathcal{E}_{\rm ChS}$ has type $\mathsf{CQ} \too \mathsf{CQ}$. (h) A Bob-with-input steering channel $\mathcal{E}_{\rm BWI}$ has type $\mathsf{CC} \too \mathsf{CQ}$. (i) An ensemble-preparing channel $\mathcal{E}_{\rm ens}$ has type $\mathsf{CC} \too \mathsf{QQ}$. (j) A quantum channel $\mathcal{E}$ has type $\mathsf{QQ} \too \mathsf{QQ}$. }\label{types}
\end{figure}

Fig.~\ref{types}(a) depicts a {\bf quantum state}, the canonical quantum resource. 
Bipartite quantum states have type $\mathsf{II} \too \mathsf{QQ}$; that is, they have no inputs and both outputs are quantum. 
The nonclassicality of quantum states is often quantified using the resource theory of local operations and classical communication (LOCC). While this is appropriate in some contexts, allowing classical communication for free is not appropriate in the context of space-like separated experiments, nor in any other scenario where distributed systems are unable to causally influence one another. In such cases, LOSR operations are the relevant ones for quantifying nonclassicality of any resource, including quantum states, and it is {\em LOSR-entanglement}, not LOCC-entanglement, that is relevant, as argued extensively in Ref.~\cite{LOSRvsLOCCentang}.

Fig.~\ref{types}(b) depicts another canonical type of resource~\cite{Barrett2005, brunner2013Bell}, often termed a correlation or a box-type resource, or {\bf box} for short. Bipartite boxes have type $\mathsf{CC} \too \mathsf{CC}$; that is, both parties have a classical input and a classical output. 
Extensive research has been done on boxes, e.g. to characterize the set of local boxes~\cite{brunner2013Bell} and the possible LOSR conversions between them~\cite{Horodecki2015, Bellquantified, Vicente2014}.
The fact that we wish to subsume boxes in our framework provides another reason to focus on LOSR as opposed to LOCC, since LOSR has been argued to be the appropriate set of free operations in this context~\cite{Bellquantified}
Furthermore, under unbounded LOCC {\em all} boxes would be deemed free, even nonlocal or signaling boxes.

Fig.~\ref{types}(c) depicts the type of resource that arises naturally in a steering scenario~\cite{Einstein1935, schrodinger_1935,wisesteer, Skrzypczyk2014, Gallego2015, Piani2015a, Cavalcanti2017, Uola2019}, often termed an {\bf assemblage}~\cite{Pusey2013}.
Such resources have type $\mathsf{CI} \too \mathsf{CQ}$; that is, the first party has a classical input and classical output, while the second party has no input and a quantum output. 

Fig.~\ref{types}(d) depicts a type of resource that arises naturally in a teleportation scenario~\cite{telep, Supic2018}, termed {\bf teleportages}~\cite{Hoban_2018}.
Such resources have type $\mathsf{QI} \too \mathsf{CQ}$.
Intuitively, given a teleportage, one would complete the standard teleportation protocol by applying one of a set of unitaries on the quantum output, conditioned on the classical output.
The precise operational sense in which these teleportages relate to the possibility of implementing an effective quantum channel is still being investigated~\cite{LipkaBartosik2019}\footnote{While LOSR is clearly the correct set of free operations for studying resources in Bell scenarios and other common cause scenarios, the same is not true for teleportation experiments, which might be better described by another resource theory (such as LOCC). The surprising insight which follows from Ref.~\cite{telep} is that a great deal can nonetheless be learned about teleportation scenarios by studying LOSR. }.

Fig.~\ref{types}(e) depicts the type of resource that arises naturally in semiquantum games, namely type $\mathsf{QQ} \too \mathsf{CC}$. We will term these  {\bf distributed measurements} or {\bf semiquantum channels}, since they arise in multiple contexts where one term~\cite{Bennett} or the other~\cite{sq} is more natural.

Fig.~\ref{types}(f) depicts the type of resource that arises naturally in measurement-device-independent (MDI) steering scenarios~\cite{steer}, namely type $\mathsf{CQ} \too \mathsf{CC}$. We will term these {\bf MDI-steering channels}.

Fig.~\ref{types}(g) depicts the type of resource that arises naturally in channel steering scenarios~\cite{channelsteer}, often termed a {\bf channel assemblage}. Such resources have type $\mathsf{CQ} \too \mathsf{CQ}$.

Fig.~\ref{types}(h) depicts the type of resource that arises when one generalizes a steering scenario to have a classical input on the steered party~\cite{BobWI}, termed a {\bf Bob-with-input steering channel}. Such resources have type $\mathsf{CC} \too \mathsf{CQ}$.

Fig.~\ref{types}(i) depicts a distributed classical-to-quantum channel, of type $\mathsf{CC} \too \mathsf{QQ}$. We will term these {\bf ensemble-preparing channels}. An interesting example of such a channel can be found in Ref.~\cite{causallocaliz} (see Eq.~82).

Fig.~\ref{types}(j) depicts a generic bipartite {\bf quantum channel}, of type $\mathsf{QQ} \too \mathsf{QQ}$.

This list is not exhaustive.
Even in the bipartite case, one might wonder how many nontrivial resource types there are, and whether all of these have been studied. First, note that the partition-type $\mathsf{I} \too \mathsf{I}$ corresponds to a trivial party. As there are no nonclassical resources involving only one party, all bipartite types involving partition-type $\mathsf{I}\too \mathsf{I}$ for either party are trivial. 
Two other partition-types, $\mathsf{C} \too \mathsf{I}$, and $\mathsf{Q} \too \mathsf{I}$, are also trivial, since the no-signaling principle guarantees that their input cannot affect the operation of the remaining parties~\cite{rosset2019characterizing}.
Moreover, some global types are equivalent up to exchange of parties, in which case we will consider only a single representative.
This leads us to our first open question. 
\begin{customopq}{1}
Even in the bipartite case, there are five nontrivial global types of resources that have not (to our knowledge) been previously studied, namely $\mathsf{QC} \too \mathsf{CQ}$, $\mathsf{CQ} \too \mathsf{QQ}$, $\mathsf{IQ} \too \mathsf{QQ}$, $ \mathsf{QQ} \too \mathsf{CQ}$, and $\mathsf{CI} \too \mathsf{QQ}$.
Do any of these correspond to scenarios which are interesting in their own right?
\end{customopq}
\noindent At the very least, each new type implies a novel form of `nonlocality'. What remains to be seen is whether these will be directly relevant for quantum information processing tasks. 

\subsection{Free  versus nonfree resources}

A nonsignaling resource (of any type) is {\bf free} with respect to LOSR, or {\bf classical}\footnote{In reference to the fact such resources can be generated by classical common causes. Classicality of a {\em resource} is not to be confused with classicality of input and output systems.}, if the parties can generate it freely using local operations and shared randomness. This notion of being free with respect to LOSR subsumes the established notions of classicality for every type of resource in Fig.~\ref{types}; e.g. for states it coincides with separability~\cite{horodecki2009quantum}, for boxes, it coincides with admitting of a local hidden variable model~\cite{brunner2013Bell}, for assemblages it coincides with unsteerability~\cite{Cavalcanti2017,Uola2019}, for teleportages it coincides with the inability to outperform classical teleportation~\cite{telep}, and so on, as pictured in Fig.~\eqref{freeset}.

\begin{figure}[htb!]
\centering
\includegraphics[width=0.48\textwidth]{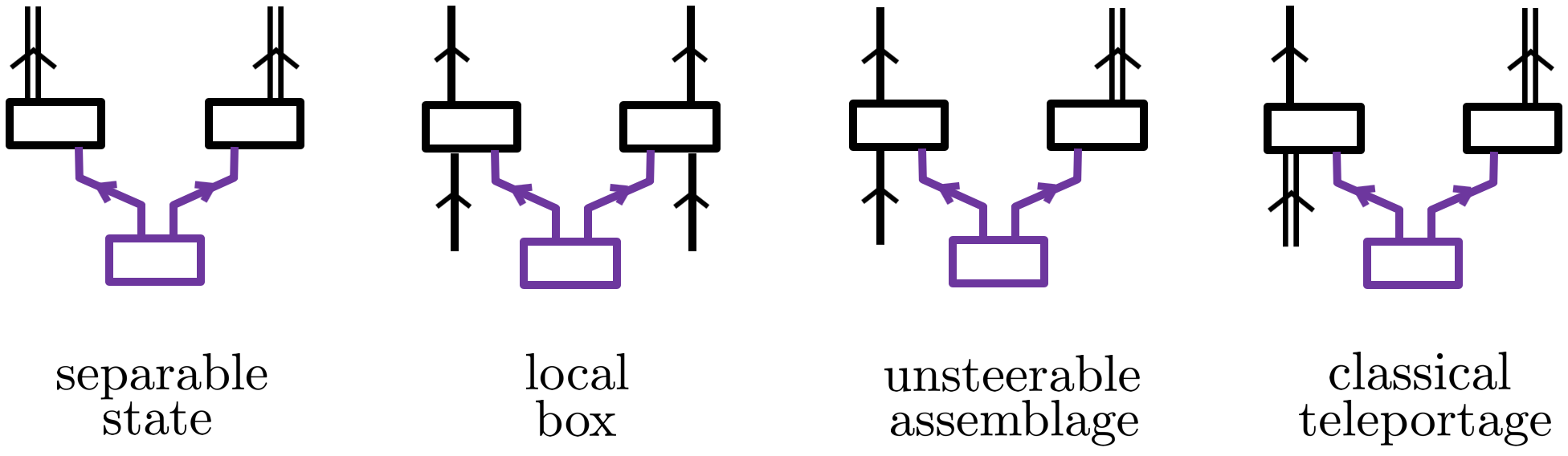}
\caption{Free LOSR resources are those which can be simulated by local operations (in black) and shared randomness (in purple). We depict four canonical types of free resources here: separable states, local boxes, unsteerable assemblages, and classical teleportages.}\label{freeset}
\end{figure}

Any resource which cannot be simulated by local operations and shared randomness is {\em non-free} and constitutes a resource of LOSR nonclassicality. 
The purpose of our type-independent resource theory of LOSR is to quantitatively characterize nonfree resources of arbitrary types, as we now do.

\subsection{Type-changing LOSR operations}

Two parties in an LOSR scenario transform resources using free LOSR operations. 
Most previous works which studied LOSR focused on conversions between specific types of resources; for example, Refs.~\cite{de2014nonlocality,gallego2016nonlocality,Bellquantified} considered LOSR conversions from boxes to boxes, Ref.~\cite{sq} considered LOSR conversions from quantum states to quantum states, and Ref.~\cite{steer} considered LOSR conversions\footnote{In this last case, the authors introduced the term local operations with steering and shared randomness (LOSSR); however, the operations they consider involve all and only the subset of LOSR operations from quantum states to assemblages, so there is no need for the new term LOSSR.} from quantum states to assemblages.
In keeping with our aim to unify a range of scenarios in one framework, and because local operations can freely change the type of a resource, we do {\em not} restrict attention to conversions among resources of fixed type, but rather allow conversions among resources of all types.

 We denote the set of all operations which can be generated by local operations and shared randomness by \LOSR. As depicted in Fig.~\ref{typechange}(a), the most general local operation on a given party is given by a comb~\cite{qcombs09}, and the different parties may correlate their choice of comb using their shared randomness. Note that this shared randomness can be transmitted down the side channel of each local comb, which implies that this depiction of LOSR is completely general and is convex~\cite{Bellquantified} for conversions from one fixed type to another. We will denote an element of this set by $\tau \in \LOSR$ and a generic resource of arbitrary type by $R$. 
 
 As in any resource theory~\cite{resthry}, the set of free operations induces a preorder over the set of all resources.  
Here, we write $R \LOSRconv R'$ whenever there exists some $\tau \in \LOSR$ such that $R'= \tau \circ R$, and we say that $R$ is {\bf at least as nonclassical} (as resourceful) as $R'$. We denote the ordering relation for the preorder defined by LOSR conversions as $\succeq_{\rm LOSR}$:
\begin{defn} \label{LOSRorder}
For resources $R$ and $R'$ of different and arbitrary type, we say that $R \succeq_{\rm LOSR} R'$ iff $R \LOSRconv R'$. 
\end{defn}
\noindent This definition allows us to make rigorous, quantitative comparisons of LOSR nonclassicality among resources of arbitrary types.
The relation $\succeq_{\rm LOSR}$ is a preorder, as there exists an identity LOSR transformation (reflexivity), and LOSR transformations compose (transitivity).

Two resources $R$ and $R'$ are equally nonclassical if they are interconvertible under LOSR; that is, if $R \LOSRconv R'$ and $R' \LOSRconv R$. We denote this $R \LOSRinterconv R'$, and we say that $R$ and $R'$ are in the same LOSR equivalence class.
 
 We give several examples of conversions among resource types in
Fig.~\ref{typechange}, depicting wires of unspecified (and arbitrary) type by dashed double lines. 
  
\begin{figure}[htb!]
\centering
\includegraphics[width=0.48\textwidth]{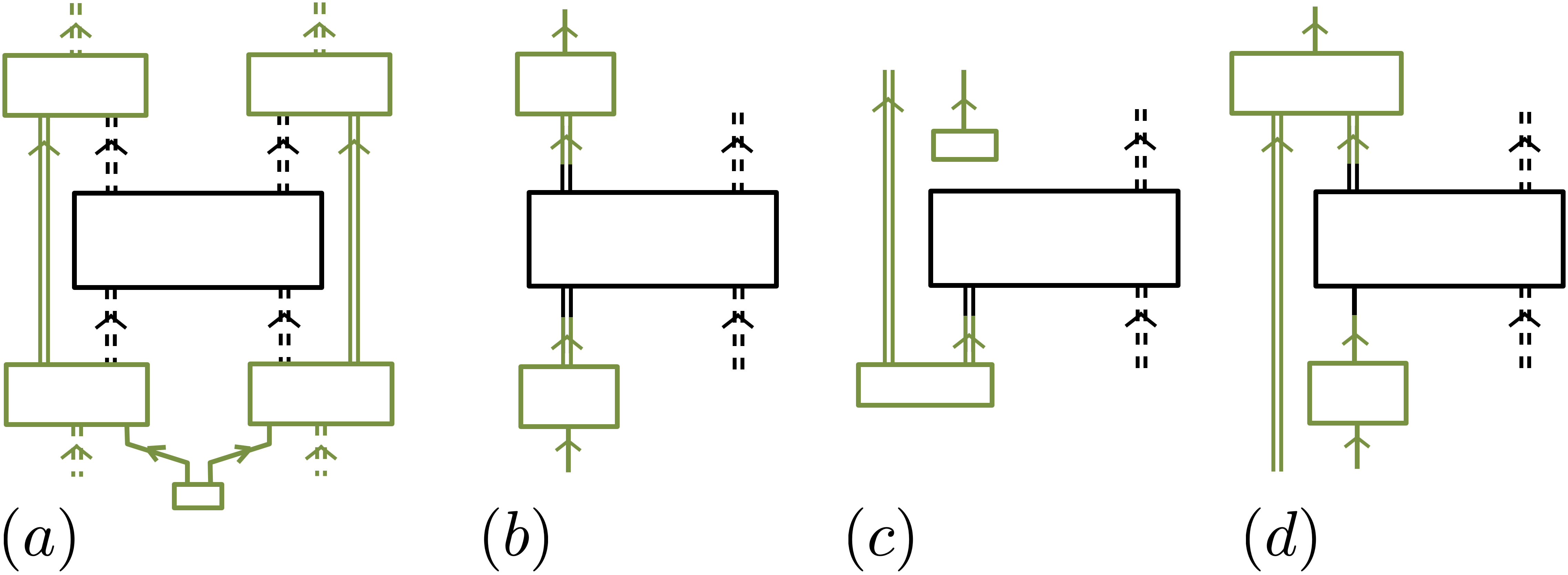}
\caption{ 
Some type-changing operations (in green), as described in the main text. Dashed wires denote systems of arbitrary and unspecified type.
(a) A generic bipartite type-changing LOSR transformation.
(b) A transformation taking partition-type $\mathsf{Q} \too \mathsf{Q}$ to $\mathsf{C} \too \mathsf{C}$.
(c) A transformation taking partition-type $\mathsf{Q} \too \mathsf{I}$ to $\mathsf{I} \too \mathsf{Q}$.
(d) A transformation taking partition-type $\mathsf{C} \too \mathsf{Q}$ to $\mathsf{Q} \too \mathsf{C}$.
} \label{typechange}
\end{figure}

Fig.~\ref{typechange}(a) depicts a generic bipartite type-changing LOSR operation.
Fig.~\ref{typechange}(b) depicts an example of a specific transformation which takes the left partition of the resource from $\mathsf{Q} \too \mathsf{Q}$ to $\mathsf{C} \too \mathsf{C}$. It is generated by composition with a local ensemble-preparing channel and a local measurement channel, respectively. 
Fig.~\ref{typechange}(c) depicts an example of a specific transformation which takes the left partition of the resource from $\mathsf{Q} \too \mathsf{I}$ to $\mathsf{I} \too \mathsf{Q}$. The transformation is generated by (sequential) composition with half of an entangled state and parallel composition with a classical system in some fixed state. In this example, the output system type is quantum, since it is comprised of a classical and quantum system.
Fig.~\ref{typechange}(d) depicts an example of a specific transformation which takes the left partition of the resource from $\mathsf{C} \too \mathsf{Q}$ to $\mathsf{Q} \too \mathsf{C}$, generated by a stochastic transformation on the classical input to the resource and performing a joint quantum measurement channel on the quantum output of the resource together with some new quantum input. 

\section{Encoding nonclassicality of one type of resource in another type} \label{expressivity}

We now consider a preorder over {\em types of resources} (rather than over the resources themselves). This allows us to formally compare the different manifestations of nonclassicality. For example, this preorder provides a {\em formal} sense in which entanglement and nonlocality are incomparable types of nonclassicality. Surprisingly, we will also show that not all types of nonclassicality are incomparable.

\begin{defn}
Global type $T$ {\bf encodes the nonclassicality of} global type $T'$, denoted $T \succeq_{\rm type} T'$, if for every resource $R'$ of type $T'$, there exists at least one resource $R$ of type $T$ such that $R' \LOSRinterconv R $.
\end{defn}
In other words, there exists some resource of the higher type in every equivalence class of resources of the lower type. Several well-known examples of such encodings will be given shortly.

To study the preorder over global types, it is also useful to consider a preorder over partition-types; that is, over the nine possible types $T_i := X_i \too Y_i$ of a single party's share of a resource.
Considering without loss of generality the first party, denoted by subscript $1$, we say that type $T_1$ is higher in the preorder than type $T_1'$ if for every resource of type $T_1'  T_2 ...  T_n $, there exists a resource of type $T_1  T_2 ...  T_n $ which is in the same LOSR equivalence class (for all numbers of parties $n$). 
Equivalently, this means that the LOSR equivalence class of any resource with partition-type $T_1'$ on the first party always contains at least one resource of partition-type $T_1$ (on the first party).
We denote this second ordering relation $\succeq_{\rm type}$:
\begin{defn} \label{typeorder}
We say that $T_1 \succeq_{\rm type} T_1'$ iff for all $R'$ of type $T_1'  T_2 ...  T_n $ (as one ranges over all $T_2, ..., T_n$ and all $n$), there exists $R$ of type $T_1  T_2 ...  T_n $ in the LOSR equivalence class of $R'$, that is, satisfying $R' \LOSRinterconv R $. 
\end{defn}
\noindent   In such cases, we say that partition-type $T_1$  encodes (the nonclassicality of) all resources of partition-type $T_1'$, or more simply that type $T_1$ encodes type $T_1'$.

If every partition-type of some given global type is higher than the corresponding partition-type of a second global type on every partition, then the first type is necessarily higher in the preorder over global types. Hence, orderings over global types can often be deduced from orderings over partition-types.

As a trivial example, it is clear that the global type $\mathsf{QQ} \too \mathsf{QQ}$ (that of bipartite quantum channels) is above every other bipartite type. For example, it is above the global type $\mathsf{II} \too \mathsf{QQ}$ (that of bipartite quantum states) in the preorder, so that $\mathsf{QQ} \too \mathsf{QQ} \succeq_{\rm type} \mathsf{II} \too \mathsf{QQ}$,  since the former is an instance of the latter where the inputs to the channel are trivial. 
In other words: given any bipartite quantum state, there is a bipartite quantum channel which is in the same LOSR equivalence class---namely, the quantum state itself, viewed as a channel from the trivial system to a quantum system on each partition. We will refer to such trivial instances of ordering among types as {\bf embeddings} of one type into the other.

Two resource types are in the same equivalence class over types if any resource of either type can be converted into a resource of the other type which is in the same LOSR equivalence class.
For example, the three partition-types $\mathsf{I}\too \mathsf{I}$, $\mathsf{C}\too \mathsf{I}$, and $\mathsf{Q}\too \mathsf{I}$ are all in the lowest equivalence class over partition-types, since (as discussed above) they never play any role in the nonclassicality of any nonsignaling resource.

Understanding the scope of nonclassicality-preserving conversions between resources of different global types is particularly useful for devising experimental measures and witnesses of nonclassicality, as we discuss in Section~\ref{implictypetoperf} (and in Ref.~\cite{rosset2019characterizing}).
Abstractly, this is because one type is above another type if there exists an embedding of the partial order over equivalence classes of resources of the lower type into the partial order of the higher type. 
When this is the case, techniques for characterizing the preorder of the higher type give direct information about the preorder of the lower type. 

\subsection{Determining which types encode the nonclassicality of which others} \label{deterencode}

In this section, we derive all but two of the ordering relations that hold between the possible pairings of partition-types by leveraging various results from the literature. These results are summarized in Table~\ref{typeorderings}. As discussed above, orderings over global types can be deduced from these.
 
\begin{table}[htb!]
\centering
\includegraphics[width=0.49\textwidth]{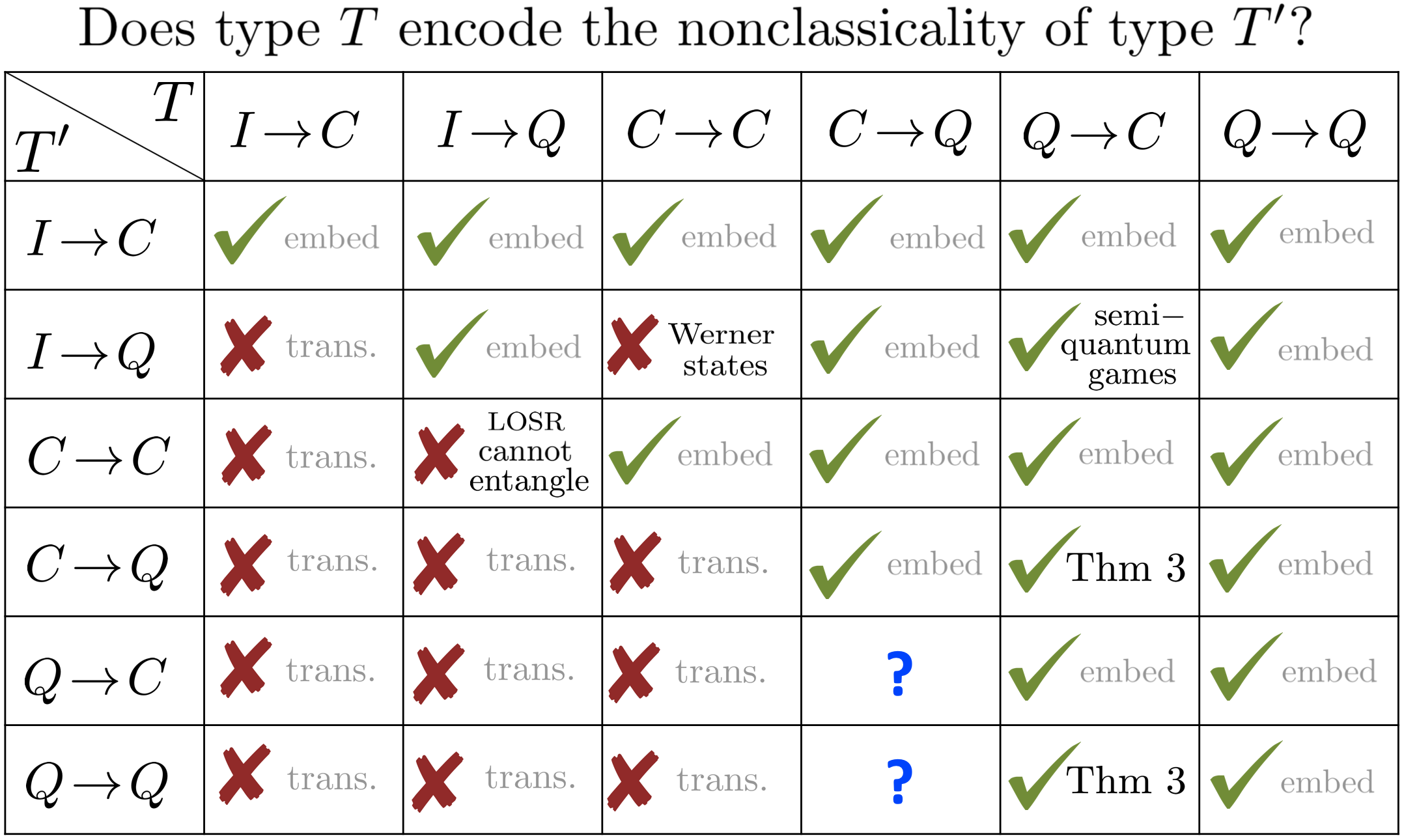}
\caption{
A green check mark in a given cell indicates that the column type $T$ is higher in the order over partition-types than the row type $T'$ (denoted $T \succeq_{\rm type} T'$), while a red cross indicates that it is not higher (denoted $T \not\succeq_{\rm type} T'$. The text in each cell alludes to the proof (given in the main text) of that ordering relation. Two relations are unknown, as indicated by blue question marks.
} \label{typeorderings}
\end{table} 
As discussed above, there are no nonfree resources which nontrivially involve the types $\mathsf{I} \too \mathsf{I}$, $\mathsf{C}    \too \mathsf{I}$, or $\mathsf{Q} \too \mathsf{I}$, so we need not discuss them further. There remain 6 nontrivial types, and hence 36 ordering relations to check.
These are all shown in the table. If the column resource type $T$  is higher in the order than the row type $T'$, so that $T \succeq_{\rm type}T'$, then we indicate this with a green check mark in the corresponding cell in the table. If instead $T \not\succeq_{\rm type}T'$, we indicate this with a red cross. In each case, we briefly allude to the logic behind the proofs for that particular ordering---proofs which we now give.

As stated in Section~\ref{expressivity}, a type is higher in the order than all types which it embeds, where quantum systems embed classical systems, which embed trivial systems. In the table, we indicate these trivial ordering relations by the word `embed'.

Next, recall that Werner proved the existence of entangled states which cannot violate any Bell inequality involving projective measurements~\cite{Werner}. It was subsequently proved that this holds true even for arbitrary local measurements~\cite{Barrettlocal}, a result that holds even if the choice of local measurements are made in a correlated fashion using shared randomness. This constitutes the most general LOSR conversion scheme from quantum states to boxes. In other words, an entangled Werner state cannot be converted into {\em any} nonfree box, much less into a box that is in its LOSR equivalence class (as would be required for encoding its nonclassicality into a box-type resource). It follows that global type $\mathsf{CC} \too \mathsf{CC}$ is not above global type $\mathsf{II} \too \mathsf{QQ}$, which in turn implies that partition-type $\mathsf{C} \too \mathsf{C}$ is not above partition-type $\mathsf{I} \too \mathsf{Q}$. That is, $\mathsf{C} \too \mathsf{C} \not\succeq_{\rm type} \mathsf{I} \too \mathsf{Q}$, as is indicated in the table by the phrase `Werner states'. 

In addition, it is well known that LOCC can generate arbitrary boxes and yet cannot generate any entangled state. Since LOSR operations form a subset of LOCC operations, this implies that LOSR operations applied to any box (of type $\mathsf{CC} \too \mathsf{CC}$) cannot generate {\em any} nonfree state (of type $\mathsf{II} \too \mathsf{QQ}$), much less a state in its LOSR equivalence class. Hence, global type $\mathsf{II} \too \mathsf{QQ}$ is not above global type $\mathsf{CC} \too \mathsf{CC}$, which in turn implies that partition-type $\mathsf{I} \too \mathsf{Q}$ is not above partition-type $\mathsf{C} \too \mathsf{C}$. That is, $\mathsf{I} \too \mathsf{Q} \not\succeq_{\rm type} \mathsf{C} \too \mathsf{C}$, as is indicated in the table by the phrase `LOSR cannot entangle'.

We can use transitivity of the ordering relation to prove that $\mathsf{I} \too \mathsf{C}$ is not above $\mathsf{I} \too \mathsf{Q}$ and is not above $\mathsf{C} \too \mathsf{C}$, and that none of $\mathsf{I} \too \mathsf{C}$, $\mathsf{I} \too \mathsf{Q}$, or $\mathsf{C} \too \mathsf{C}$ are above any of $\mathsf{C} \too \mathsf{Q}$, $\mathsf{Q} \too \mathsf{C}$, and $\mathsf{Q} \too \mathsf{Q}$. For example, from the fact that $\mathsf{C} \too \mathsf{C}$ is above $\mathsf{I} \too \mathsf{C}$ and the fact that $\mathsf{C} \too \mathsf{C}$ is not above $\mathsf{I} \too \mathsf{Q}$, it must be that $\mathsf{I} \too \mathsf{C}$ is not above $\mathsf{I} \too \mathsf{Q}$. If it were otherwise, one would have $\mathsf{C} \too \mathsf{C}$ above $\mathsf{I} \too \mathsf{C}$ above $\mathsf{I} \too \mathsf{Q}$ $\implies$ $\mathsf{C} \too \mathsf{C}$ above $\mathsf{I} \too \mathsf{Q}$, which is false. The other transitivity arguments run analogously. In the table, we indicate all such ordering relations by the abbreviation `trans.'.

One of the authors proved in Ref.~\cite{sq} that there exists some semiquantum channel (of type $\mathsf{QQ} \too \mathsf{CC}$) in the same equivalence class as any given quantum state (of type $\mathsf{II} \too \mathsf{QQ}$). A slight reframing of this result implies that the semiquantum partition-type $\mathsf{Q} \too \mathsf{C}$ is higher in the order than $\mathsf{I} \too \mathsf{Q}$, as we show below. That is, $\mathsf{Q} \too \mathsf{C} \succeq_{\rm type} \mathsf{I} \too \mathsf{Q}$, as is indicated in the table by the phrase `semiquantum games'.

Finally, as we prove in Theorem~\ref{squniv}, the semiquantum partition-type $\mathsf{Q} \too \mathsf{C}$ is higher in the order than all other partition-types. The ordering relations that follow from our proof but not from previous work, namely $\mathsf{Q} \too \mathsf{C} \succeq_{\rm type} \mathsf{C} \too \mathsf{Q}$ and $\mathsf{Q} \too \mathsf{C} \succeq_{\rm type} \mathsf{Q} \too \mathsf{Q}$, are indicated in the table by the phrase `Thm~3'.

This proves all the results shown in the table. There remain two unknown ordering relations, indicated in the table by question marks; namely whether $\mathsf{C} \too \mathsf{Q}$ is higher in the order than either $\mathsf{Q} \too \mathsf{C}$ or $\mathsf{Q} \too \mathsf{Q}$. Because $\mathsf{Q} \too \mathsf{C}$ and $\mathsf{Q} \too \mathsf{Q}$ are in the same equivalence class (at the top of the order), the answer to both of these questions must be the same; that is, either $\mathsf{C} \too \mathsf{Q}$ encodes them both, or it encodes neither.
Such an encoding could have dramatic practical consequences. 
For example, if the encoding can be done with a fixed transformation (which is not a function of the resource to be converted), then this would enable the possibility of {\em preparation-device-independent} quantification of nonclassicality.
\begin{customopq}{2}
Can the LOSR nonclassicality of any resource be perfectly characterized in a {\em preparation-device-independent} manner? 
\end{customopq}

\subsection{Semiquantum channels are universal encoders of nonclassicality} \label{sec:sq}

To complete the arguments of the last section, we prove that the semiquantum partition-type can encode any other partition-type. The consequences of this fact are fleshed out further in Section~\ref{implictypetoperf}.
\begin{thm} \label{squniv}
The semiquantum partition-type $\mathsf{Q} \too \mathsf{C}$ is in the unique equivalence class at the top of the order over partition-types. That is, it can encode the nonclassicality of all other partition-types.
\end{thm}
\begin{proof}
Consider a bipartite channel $\mathcal{E}$ which has a quantum output of dimension $d$, together with arbitrary other outputs and inputs (denoted by dashed double lines), as shown in black in Fig.~\ref{SQandBack}(a). One can transform $\mathcal{E}$ into a resource with a quantum input of dimension $d$ and a classical output of dimension $d^2$ by composing $\mathcal{E}$ with a Bell measurement as shown in green in Fig.~\ref{SQandBack}(a); that is, by performing a measurement in a maximally entangled basis on the quantum output of $\mathcal{E}$ and a new quantum input of the same dimension $d$.
\begin{figure}[htb!]
\centering
\includegraphics[width=0.48\textwidth]{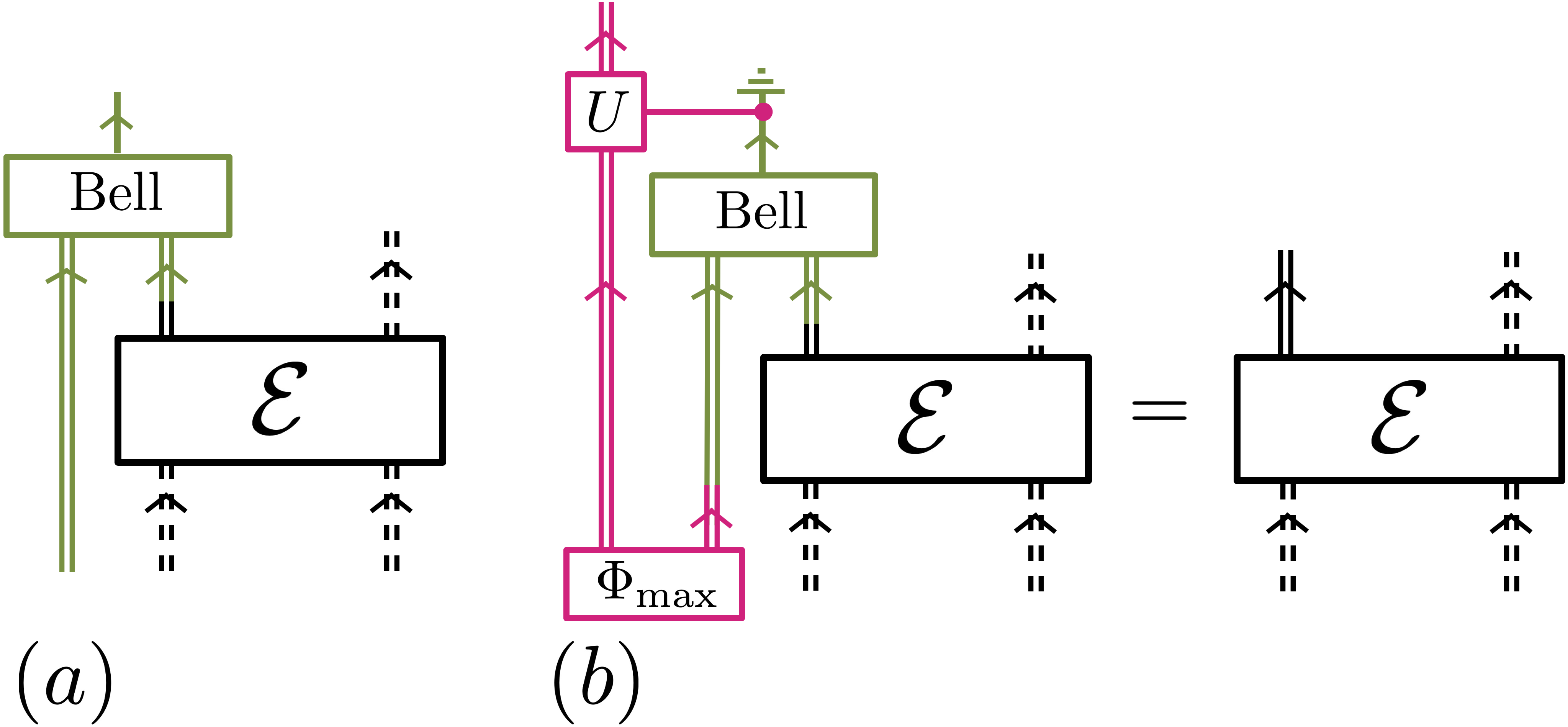}
\caption{  (a) A free transformation (in green) that converts a quantum output to a classical output together with a new quantum input. 
(b) This transformation does not change the LOSR equivalence class, since it has a left inverse (shown in pink) which is a free transformation.} \label{SQandBack}
\end{figure}
To see that this transformation preserves LOSR equivalence class, it suffices to note that there exists a local (and hence free) operation, shown in pink on the left-hand side of Fig.~\ref{SQandBack}(b), which takes the transformed channel back to the original channel $\mathcal{E}$. In particular, this local operation feeds one half of a maximally entangled state $\Phi_{\rm max}$ into the Bell measurement, and then performs a correcting unitary operation $U$ on the other half of the entangled state, conditioned on the classical outcome of the Bell measurement. For the correct choice of correction operations, the overall transformation on $\mathcal{E}$ is just the well-known teleportation protocol~\cite{Bennett93}, and so the equality shown in Fig.~\ref{SQandBack}(b) holds. Hence, the channel in Fig.~\ref{SQandBack}(a) is in the same LOSR equivalence class as $\mathcal{E}$, which implies that every partition of a resource can be transformed to a resource of type $\mathsf{Q} \too \mathsf{C}$ in the same equivalence class.
\end{proof}

Note that $\mathsf{Q} \too \mathsf{Q}$ is trivially also at the top of the order, since every other type embeds into it. It is thus in the same equivalence class as $\mathsf{Q} \too \mathsf{C}$.

\section{A unified framework for distributed games of all types }
\label{unifiedgames}

A variety of `games' have been studied for the purposes of quantifying nonclassicality of various types of resources. 
For instance, the nonclassicality of quantum states has been studied from the point of view of nonlocal games and semiquantum games, as well as teleportation, steering, and entanglement witnessing experiments. Nonlocal games have also been used to study the nonclassicality of boxes. 

In fact, there is a natural class of distributed tasks for every type of resource, including one for each of the common types in Section~\ref{sectypes}. 
\begin{defn} \label{defn:Tgame}
For a given global type $T$, we define a distributed $\mathbf{T}${\bf -game} as a linear map from resources of type $T$ to the real numbers. 
\end{defn}
The set $\mathcal{G}_{\! T}$ of all such maps for fixed $T$ is the set of $T$-games, and a resource of type $T$ is said to be a {\bf strategy} for a $T$-game. This last terminology is motivated by the fact that no matter how complicated the players' tactics, their score for a given $T$-game only depends on the resource of type $T$ that they ultimately share with the referee. We will refer to any game of any type as a distributed game.

In Fig.~\ref{games}, we depict four distributed games together with the type of resource that acts as a strategy for that game. We represent a game diagrammatically as a monolithic comb with appropriate input and output structure such that composition of the comb corresponding to a game $G_{\! T}$ with a strategy $\mathcal{E}_T$ of type $T$ yields a circuit with no open inputs or outputs, representing the real number $G_{\! T}(\mathcal{E}_T)$. 

\begin{figure}[htb!]
\centering
\includegraphics[width=0.48\textwidth]{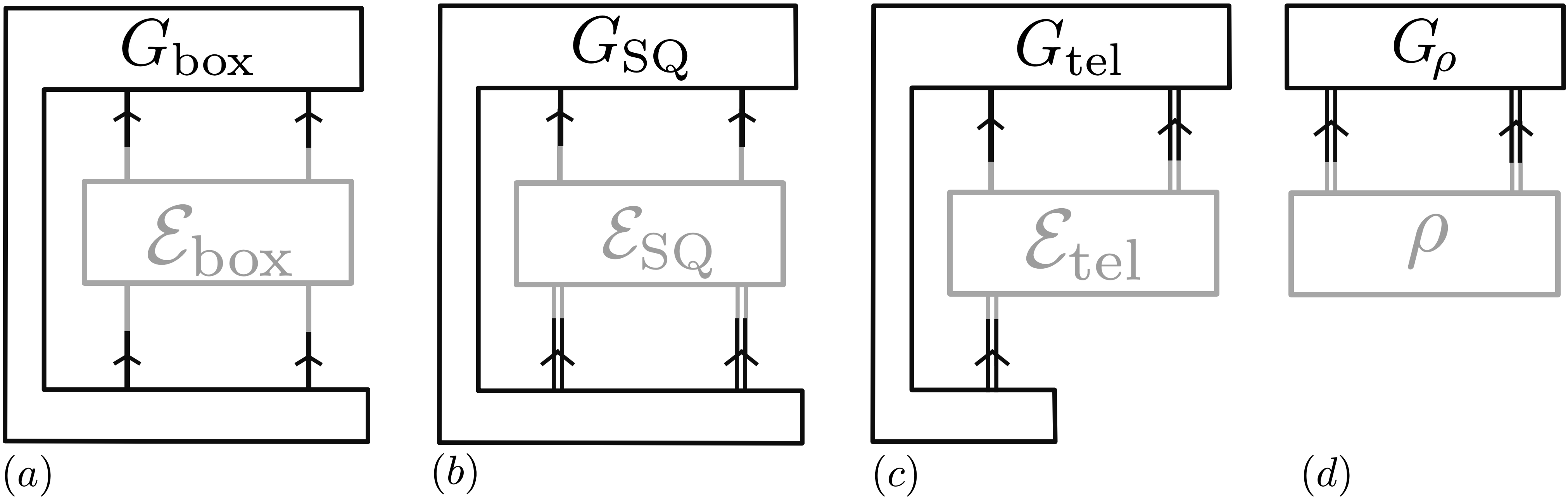}
\caption{ Some games and their strategies. (a) Boxes are strategies for nonlocal games. (b) Semiquantum channels are strategies for semiquantum games. (c) Teleportages are strategies for teleportation games. (d) Entangled states are strategies for entanglement witnesses. 
} \label{games}
\end{figure}

\subsection{Implementations of a game}

We have noted that a variety of games and experiments can be viewed abstractly under the umbrella of $T$-games. The {\em practical} meaning of such games is made more clear by considering the following two-step procedure, by which a referee can implements any game (of any type $T$). This procedure is depicted on the right-hand side of Fig.~\ref{TCgames}. 

First, the referee performs a tomographically complete measurement on the composite system defined by the collection of output systems of the given strategy $\mathcal{E}_T$, and implements a preparation drawn at random from a tomographically complete set of preparations on the composite system defined by the collection of all the systems which are inputs of $\mathcal{E}_T$. 
In fact, it suffices for the referee to perform tomographically complete measurements and preparations {\em independently} on every input and output, as depicted in the dashed box in Fig.~\ref{TCgames}. 
We will refer to this process as the application of an {\bf analyzer} $Z$ to the given strategy. That is, an analyzer $Z$ is a linear and tomographically complete map from strategies to correlations of the form $P_{Z\circ \mathcal{E}_T}(ab|xy) := Z\circ \mathcal{E}_T$, with $a,b$ labeling the values of the classical outputs of $Z$ and $x,y$ the values of the classical inputs of $Z$. Second, the referee uses a fixed payoff function $F_{\rm payoff}(abxy)$ to assign a real number  $G_{\! T}(\mathcal{E}_T)= \sum_{abxy} F_{\rm payoff}(abxy) P_{Z\circ \mathcal{E}_T}(ab|xy)$  to strategy $\mathcal{E}_T$. 


This point of view on games is useful for the proof of Theorem~\ref{subsumestrat}, and it is also useful for establishing a physical picture of games of each type. For example, in a Bell experiment, one applies LOSR operations (or often just LO operations) in order to convert one's quantum state to a conditional probability distribution, and the payoff function in the game constitutes the Bell inequality that one tests. As a second example, see Ref.~\cite{lipkabartosik2019operational} for a study of various teleportation games. As noted therein, there are interesting teleportation tasks (which admit of a simple operational interpretation) beyond merely attempting to establish an identity channel between two parties using shared entanglement.
However, in the rest of this paper it will be simpler to view a game in the abstract (simply as a linear map from resources of a given type to the real numbers), and we will leave the further investigation of such games (beyond the cases which have already been studied) to future work.

\begin{figure}[htb!]
\centering
\includegraphics[width=0.45\textwidth]{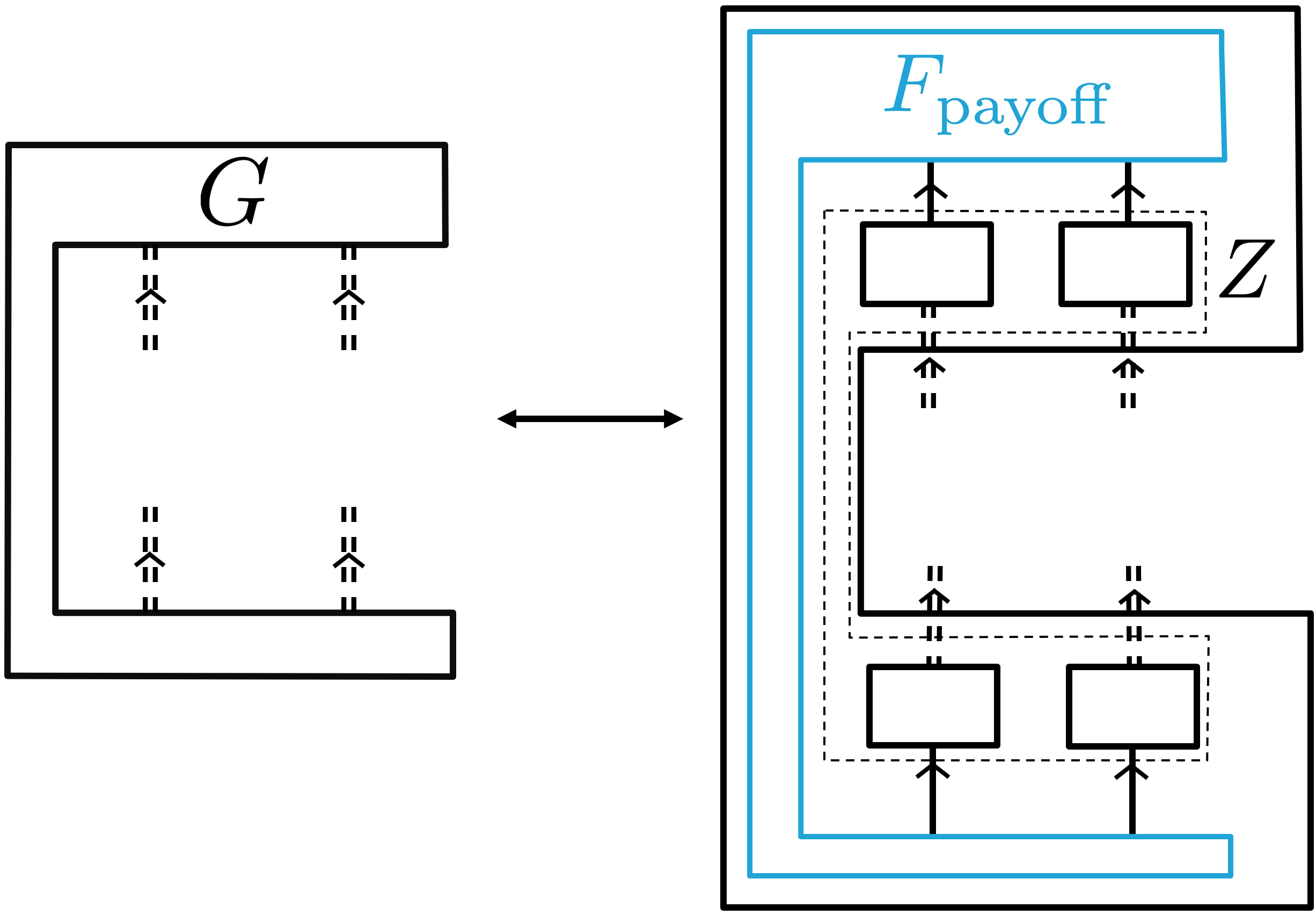}
\caption{
A depiction of the concrete two-step process by which a referee can implement a game (of any type). The referee first applies a tomographically complete analyzer $Z$, and then assigns a real number to the resulting statistics using a payoff function $F_{\rm payoff}$.
} \label{TCgames}
\end{figure}

\subsection{Performance of resources of arbitrary type with respect to a game}

By definition, every $T$-game assigns a real number to every resource of type $T$. At this stage, the number need not be related in any way to the nonclassicality of resources; e.g., the score need not behave monotonically under LOSR. Nonetheless, one can use any $T$-game to learn about the LOSR ordering of resources of type $T$; indeed, the full set of $T$-games perfectly characterizes this preorder. (In case this is not completely obvious, it will follow as a corollary of our Theorem~\ref{proporder}.) 
Furthermore, one can use a $T$-game to (partially) quantify the nonclassicality of a resource {\em of arbitrary type}, not only of type $T$. For example, nonlocal games and semiquantum games have been used to probe the nonclassicality of quantum states~\cite{sq,Shahandeh2017,Supic2017,Rosset2018a}. 
 
This is because---although a $T$-game does not {\em directly} assign a score to
resources of any type other than $T$---it can quantify the performance of a resource of any type by a maximization over all $\tau \in {\rm LOSR}$ which map the given resource to one of type $T$. That is:
\begin{defn}
The (optimal) {\em performance} of a resource $R$ of arbitrary type with respect to a game $G_{\! T}$ of arbitrary type $T$ is given by
\beq
\omega_{G_{\! T}}(R) = \max_{\tau: \mathbf{Type}[R] \! \to \! T} G_{\! T}(\tau \circ R).
\eeq
\end{defn}
Clearly, $\omega_{G_{\! T}}(R)$ is a measure of how well an arbitrary resource $R$ can perform at LOSR-game $G_{\! T}$. Because of the maximization over LOSR operations, $\omega_{G_{\! T}}(R)$ is by construction a monotone with respect to LOSR. Constructions of this sort are often termed {\em yield monotones}~\cite{Gonda2019}. 
We discuss monotones further in Ref.~\cite{rosset2019characterizing}, as monotones are useful tools for obtaining partial information about the preorder over resources and for relating the preorder to practical tasks.

The set $\mathcal{G}_{\! T}$ of all games of a given type $T$ defines a preorder over all resources of all types, where resource $R$ is above $R'$ in the order if for every $T$-game, $R$ can achieve a value at least as high as $R'$ can.
We denote this third ordering relation $\succeq_{\mathcal{G}_{\! T}}$:
\begin{defn} \label{orderwrtgame}
For resources $R$ and $R'$ of different and arbitrary type, we say that $R \succeq_{\mathcal{G}_{\! T}} R'$ iff $\omega_{G_{\! T}}(R) \ge \omega_{G_{\! T}}(R')$ for every $G_{\! T} \in \mathcal{G}_{\! T}$. 
\end{defn}

Next, we prove that if one resource outperforms a second at all possible games of a given type, then it can also generate {\em any specific strategy} of that type which the second resource can generate. This is a nontrivial result, since it need not be the case that the first resource is higher in the LOSR order.
\begin{thm} \label{subsumestrat}
For resources $R$ and $R'$ of different and arbitrary type and a resource $\mathcal{E}_{T}$ of arbitrary type $T$, 
$R \succeq_{\mathcal{G}_{\! T}} R'$ iff $R' \LOSRconv \mathcal{E}_{T} \implies R \LOSRconv \mathcal{E}_{T}$. That is, any strategy $\mathcal{E}_{T}$ for games of type $T$ that can be freely generated from $R'$ can {\em also} be freely generated from $R$. 
\end{thm}

\begin{proof}
If $R' \LOSRconv \mathcal{E}_{T} \implies R \LOSRconv \mathcal{E}_{T}$, then $R$ can generate any strategy for any given game $G_{\! T}$ that $R'$ can, and so always performs at least as well as $R'$ at $T$-games, and so $R \succeq_{\mathcal{G}_{\! T}} R'$.

To prove the converse, consider a set of games of type $T$ defined by ranging over all possible payoff functions $F_{\rm payoff}(abxy)$ for some fixed analyzer $Z$---that is, a specific tomographically complete measurement for each output system of the resource and a specific tomographically complete set of states for each input system of the resource.
Assume that $R' \LOSRconv \mathcal{E}_{T}$ for some strategy $\mathcal{E}_{T}$, and define $P_{Z\circ \mathcal{E}_T}(ab|xy) = Z \circ \mathcal{E}_{T}$. 
For $R \succeq_{\mathcal{G}_{\! T}} R'$, it must be that $R \LOSRconv \mathcal{E}_{T}'$ for at least one strategy $\mathcal{E}_{T}'$ satisfying $P_{Z\circ \mathcal{E}'_T}(ab|xy) = Z \circ \mathcal{E}_{T}'$. If this were {\em not} the case, then the convex set $S(R)$ of all correlations which $R$ can generate in this scenario, $S(R):=\left\{P_{Z \circ \tau \circ R}(ab|xy)=Z \circ \tau \circ R \right\}_{\tau \in {\rm LOSR}}$, would not contain $P_{Z\circ \mathcal{E}_T}(ab|xy)$, and the hyperplane which separated $P_{Z\circ \mathcal{E}_T}(ab|xy)$ from $S$ would constitute a payoff function $F_{\rm payoff}$ for which $R'$ outperformed $R$, which would be in contradiction with the claim that $R \succeq_{\mathcal{G}_{\! T}} R'$. By tomographic completeness, the preimage of every correlation  under $Z$ contains at most one strategy. Hence, if two strategies map to the same correlation, then they must be the same strategy, and so it must be that $\mathcal{E}_{T}=\mathcal{E}_{T}'$ in argument above. That is, we have shown that if $R \succeq_{\rm SQ} R'$ and $R' \LOSRconv \mathcal{E}_{T}$, then $R \LOSRconv \mathcal{E}_{T}$.
\end{proof}

\subsection{Implications from the type of a resource to its performance at games} \label{implictypetoperf}

We now prove that games of a higher type perfectly characterize the LOSR nonclassicality of resources of a lower type. 
\begin{thm} \label{proporder}
If $T \succeq_{\rm type} T'$, then for resources $R_1,R_2$ of type $T'$, $R_1 \succeq_{\rm LOSR} R_2$ iff $R_1 \succeq_{\mathcal{G}_{\! T}} R_2$. Equivalently: if type $T$ is above type $T'$, then for resources of type $T'$, the orders defined by $\succeq_{\rm LOSR}$ and $\succeq_{\mathcal{G}_{\! T}}$ are identical.
\end{thm}
\begin{proof}
Consider the set $\mathcal{G}_{\! T}$ of all games of type $T$ and two resources $R_1$ and $R_2$, both of type $T'$, where $T \succeq_{\rm type} T'$.
Clearly $R_1 \succeq_{\rm LOSR} R_2$ implies $R_1 \succeq_{\mathcal{G}_{\! T}} R_2$, since $R_1 \succeq_{\rm LOSR} R_2$ implies that $R_1$ can be used to freely generate $R_2$ and hence to generate any strategy which can be generated using $R_2$.
Next, we prove that $R_1 \succeq_{\mathcal{G}_{\! T}} R_2$ implies $R_1 \succeq_{\rm LOSR} R_2$. By assumption, $T \succeq_{\rm type} T'$, and so for $R_2$ of type $T'$, there exists a strategy $\mathcal{E}_{ T}$ for games of type $T$ such that $R_2 \LOSRinterconv \mathcal{E}_{ T}$.
Since $R_1 \succeq_{\mathcal{G}_{\! T}} R_2$, Theorem~\ref{subsumestrat} tells us that $R_2 \LOSRconv \mathcal{E}_{ T}$ implies $R_1 \LOSRconv \mathcal{E}_{ T}$, and hence $R_1 \LOSRconv R_2$ by transitivity.
Hence we have proven that the two orderings are the same; that is, $R_1 \succeq_{\rm LOSR}R_2$ if and only if $R_1 \succeq_{\mathcal{G}_{\! T}}R_2$. 
\end{proof}

A consequence of this result is that if $T \succeq_{\rm type} T'$, then every nonfree resource of type $T'$ is useful for some $T$-game. Two special cases of this fact that were previously proved are that all entangled states are useful for semiquantum games and that all entangled states are useful for teleportation. 

If one views the encoding of one type into another type as an embedding of the partial order over equivalence classes of resources of the lower type into the partial order of the higher type, then this result can be seen as a consequence of the fact that games of type $T$ are sufficient for characterizing the partial order over resources of type $T$.

A corollary of Theorem~\ref{squniv} and Theorem~\ref{proporder} is that semiquantum games fully characterize the LOSR ordering among all resources of arbitrary type.
\begin{cor} \label{Buscgen}
For any resources $R$ and $R'$ (which may be of arbitrary and different types), $R \succeq_{\rm LOSR} R'$ if and only if $R \succeq_{\mathcal{G}_{\rm SQ}}R'$. 
\end{cor}
\noindent This generalizes the main result of Ref.~\cite{sq} from quantum states to resources {\em of arbitrary type}.
Since semiquantum games characterize the LOSR nonclassicality of arbitrary resources, and since referees in semiquantum games do not require any well-characterized quantum measurement devices~\cite{steer}, it follows that the nonclassicality of any resource of any type can be characterized in a measurement-device-independent manner. 

Note that for such tests to be practically useful, it must be possible to convert an {\em unknown} resource into a semiquantum channel in the same LOSR equivalence class. This is indeed possible, because for all resources of a given type, there is a {\em single} transformation which implements the conversion, namely, the Bell measurement in Fig.~\ref{SQandBack}(a). Critically, this transformation is not a function of the resource to be converted.

\section{Extending results from the literature }
\label{extending}
We now give further applications of our results, in particular showing how our framework extends a number of seminal results from the literature.

\subsection{Applying semiquantum games to perfectly characterize arbitrary quantum channels} 

Buscemi proved in Ref.~\cite{sq} that the order over quantum states with respect to LOSR is equivalent to the order over quantum states defined by their performance with respect to semiquantum games. This result is an instance of our Corollary~\eqref{Buscgen} where $R$ and $R'$ are both quantum states. 

For concreteness, we now briefly reiterate the argument in this specific context. The existence of the invertible transformation in Fig.~\ref{SQandBack} implies that $\mathsf{II} \too \mathsf{QQ}$ is below $\mathsf{QQ} \too \mathsf{CC}$ in the order on global types, and hence that

\begin{figure}[h!]
\centering
\includegraphics[width=0.47\textwidth]{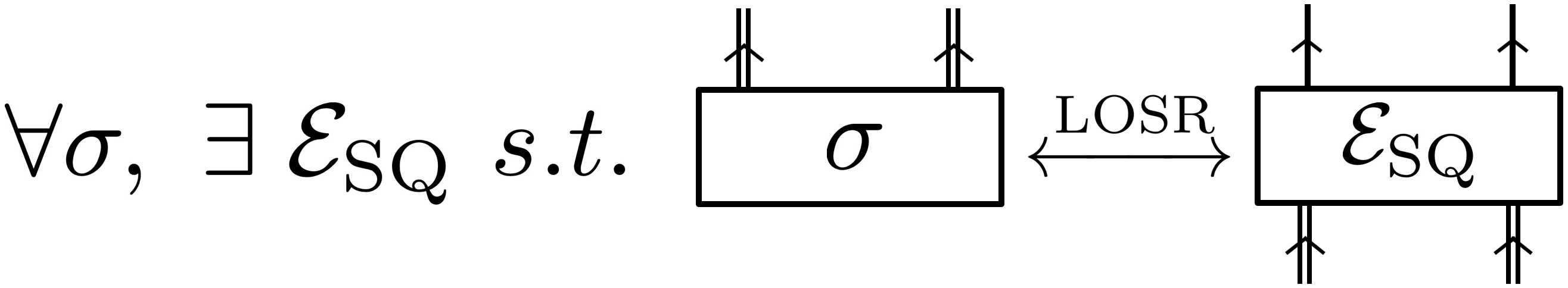}.
\end{figure}

\noindent For this $\sigma$ and $\mathcal{E}_{\rm SQ}$ such that $\sigma \LOSRconv \mathcal{E}_{\rm SQ}$, Theorem~\ref{subsumestrat} states that if $\rho \succeq_{\mathcal{G}_{\rm SQ}} \sigma$, then $\rho \LOSRconv \mathcal{E}_{SQ}$. Since $\mathcal{E}_{\rm SQ} \LOSRconv \sigma$, transitivity gives that $\rho \succeq_{\mathcal{G}_{\rm SQ}} \sigma \implies \rho \LOSRconv \sigma$. 
Since the converse implication is self-evident, one sees that the LOSR order over quantum states is equivalent to the order over quantum states defined by their performance with respect to semiquantum games.

This proof is inspired by the original argument in Ref.~\cite{sq}, but our framework makes the proof shorter and more intuitive. As we saw in Corollary~\eqref{Buscgen}, it also allowed us to generalize the result from quantum states to arbitrary resources.
As stated above, this implies that the LOSR nonclassicality of {\em any} resource can be witnessed and quantified in a measurement-device-independent~\cite{Branciard2013,steer} manner.

\subsection{Applying measurement-device-independent steering games to perfectly characterize assemblages } 

Cavalcanti, Hall, and Wiseman proved in Ref.~\cite{steer} that the LOSR order over quantum states defined by subset inclusion over the assemblages that each can generate via LOSR is equivalent to the order over quantum states defined by their performance with respect to steering games.  
This result is a special case of our Theorem~\ref{subsumestrat}, where $R$ and $R'$ are quantum states and $\mathcal{E}_T$ is a steering assemblage:
\begin{cor}
 $\rho \succeq_{\mathcal{G}_{\rm steer}} \sigma$ iff $\sigma \LOSRconv \mathcal{E}_{\rm steer} \implies \rho \LOSRconv \mathcal{E}_{\rm steer}$. 
\end{cor}
\noindent Our Theorem~\ref{subsumestrat} extends this result to arbitrary resource types and games.

Additionally, the existence of the invertible transformation in Fig.~\ref{SQandBack} immediately implies that

\begin{figure}[htb!]
\centering
\includegraphics[width=0.47\textwidth]{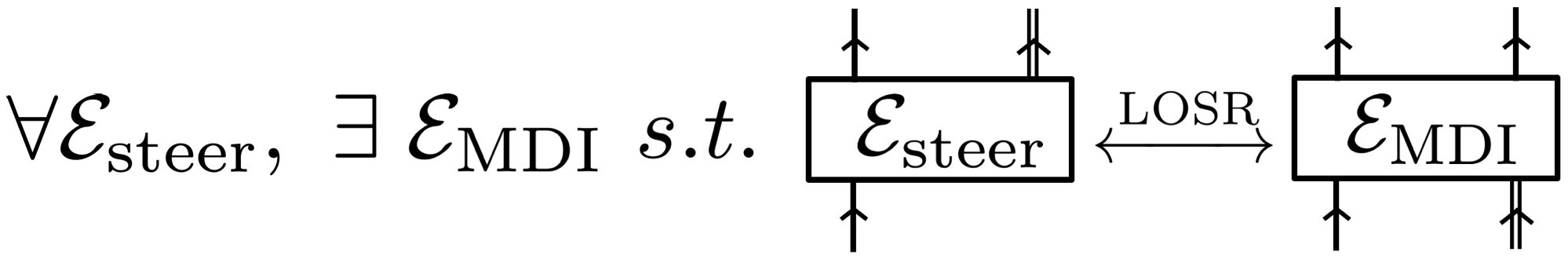}.
\end{figure}
\noindent In other words, $\mathsf{CI} \too \mathsf{CQ}$ is below $\mathsf{CQ} \too \mathsf{CC}$ in the order on global types.
Our Theorem~\ref{proporder} then gives a new result, which is the direct analogue of the result in Ref.~\cite{sq} in this new context: the LOSR order over assemblages is equivalent to the order over assemblages defined by performance relative to all measurement-device-independent steering games.
Explicitly: the fact (proven in Section~\ref{deterencode}) that $T_{\rm MDI} \succeq_{\rm type} T_{\rm steer}$  implies that
\begin{cor}
For two assemblages $\mathcal{E}_{\rm steer}$ and $\mathcal{E}_{\rm steer}'$, one has $\mathcal{E}_{\rm steer} \succeq_{\rm LOSR} \mathcal{E}_{\rm steer}'$ iff $\mathcal{E}_{\rm steer} \succeq_{\mathcal{G}_{\rm MDI}} \mathcal{E}_{\rm steer}'$.
\end{cor}
\noindent Indeed, this theorem holds not just for assemblages, but for any resource type which is lower in the global order than measurement-device-independent steering channels, including channel steering assemblages and Bob-with-input assemblages.

\subsection{Applying teleportation games to perfectly characterize quantum states} \label{app:tel}

Cavalcanti, Skrzypczyk, and $\check{\rm S}$upi\'{c} proved in Ref.~\cite{telep} that the nonclassicality of every entangled state can be witnessed by some teleportation experiment. 
We apply arguments analogous to those of the last two subsections to strengthen their results, most notably in Corollary~\ref{quantup}, which provides the quantitative analogue of their (qualitative) main result.

First, the existence of the invertible transformation in Fig.~\ref{SQandBack} again implies that

\begin{figure}[htb!]
\centering
\includegraphics[width=0.47\textwidth]{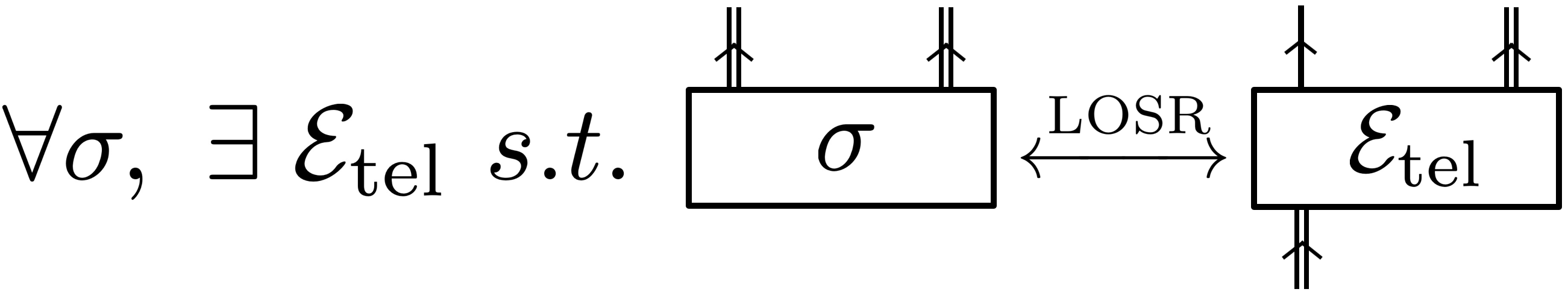}.
\end{figure}

\noindent In other words, $\mathsf{II} \too \mathsf{QQ}$ is below $\mathsf{QI} \too \mathsf{CQ}$ in the order on global types.
Our Theorem~\ref{proporder} again yields a result analogous to that in Ref.~\cite{sq}, namely, that the LOSR order over entangled states is equivalent to the order over entangled states with respect to performance at teleportation games\footnote{It is worth noting that there are subtleties in the relationship between teleportation games (as defined here, and see also Ref.~\cite{lipkabartosik2019operational}) and the usual conception of teleportation experiments (as attempts to establish an identity channel between two parties using shared entanglement). For example, note that any nonfree assemblage constitutes a special instance of a teleportage which is useless for generating a coherent quantum channel between two parties, and yet which {\em is} useful for some teleportation game.}.
Explicitly: denoting the type of quantum states by $T_{\rho}$, the fact that $T_{\rm tel} \succeq_{\rm type} T_{\rho}$  implies that
\begin{cor} \label{quantup}
$\rho \succeq_{\rm LOSR} \sigma$ iff $\rho \succeq_{\mathcal{G}_{\rm tel}} \sigma$.
\end{cor}
\noindent Indeed, this theorem holds not just for quantum states, but for any resource type which is lower in the global order than teleportages, including, for example, steering assemblages.

Our Theorem~\ref{subsumestrat} can also be applied to teleportation games, yielding a result analogous to that in Ref.~\cite{steer}. That is, any resource which outperforms a second resource at all teleportation games can generate any specific strategy that the second can generate:
\begin{cor}
 $R \succeq_{\mathcal{G}_{\rm tel}} R'$ iff $R' \LOSRconv \mathcal{E}_{\rm tel} \implies R \LOSRconv \mathcal{E}_{\rm tel}$.
\end{cor}

\section{Open questions}

Our framework suggests a great deal of open questions for future study, two important examples of which were highlighted above.

Ideally, one would have type-independent methods for characterizing nonclassicality in practice. We begin developing such a toolbox in Ref.~\cite{rosset2019characterizing}. 

For each of the fifteen bipartite global types mentioned above, it is interesting to study the basic features of the (type-specific) LOSR resource theory. While this has been done for boxes, little attention has been given to this problem in other cases, even for quantum states.

Such features include the geometry of the free set of resources, the LOSR preorder, useful monotones and witnesses, and so on.  Ultimately, we advocate not just for these type-specific investigations, but for research in the type-independent context.

Part of the project of characterizing the preorder will be to characterize the sense in which there exist inequivalent kinds of nonclassicality.
At the top of the preorder, the situation for bipartite LOCC-entanglement is quite simple: there is a single maximally entangled state of a given dimension, from which all other states can be obtained by LOCC transformations.
This is no longer the case for multipartite LOCC-entanglement~\cite{Dur2000twoways}, nor for LOSR-entanglement even in the bipartite case~\cite{LOSRvsLOCCentang}. For resources beyond quantum states and for more parties, the situation gets even more complex. As an example, our work implies that there exist semiquantum channels in the equivalence class of Werner states, and semiquantum channels in the equivalence class of nonlocal boxes, and that these semiquantum channels exhibit inequivalent forms of nonclassicality.
\begin{customopq}{3}
What are the key features of the type-independent preorder over LOSR resources? What inequivalent forms of nonclassicality do these resources exhibit?
\end{customopq}

If one was interested only in {\em witnessing} nonclassicality as opposed to {\em quantifying} it, one could consider a preorder over types defined by a less restrictive condition, where type $T$ is above type $T'$ if every nonfree resource of type $T'$ could freely generate at least one nonfree resource of type $T$. All the known results in Table~\ref{typeorderings} hold for this definition as well; however, the two definitions might yield different answers for the open questions that remain.

One could also consider modifying our Definition~\ref{typeorder} such that local operations were taken to be free rather than local operations and shared randomness. Note that the operations required in the proof of Theorem~\ref{SQandBack} do not make use of any shared randomness, and so the theorem would still hold. In fact, one can readily verify that all the orderings in Figure~\ref{typeorderings} would continue to hold. However, Theorem~\ref{subsumestrat} requires convexity (through its use of the separating hyperplane theorem), as do Theorem~\ref{proporder} and Corollary~\ref{Buscgen} (since they rely on Theorem~\ref{subsumestrat}). 

If one were to modify Definition~\ref{typeorder} so that local operations and classical communication were free, the situation is less clear, as one would presumably need to widen the scope of applicability to signaling resources.
\begin{customopq}{4}
What can be learned by considering a type-independent framework of LOCC nonclassicality? 
\end{customopq}
This would be the relevant resource theory, for example, for distributed parties who share quantum memories and the ability to communicate classically.

Our framework has focused on the divide between classical and quantum resources. One can also study the divide between quantum and post-quantum resources, as we do in Ref.~\cite{postquantumschmid}.

A final open question regards the relationship between our work and self-testing~\cite{Mayers2004,Montanaro2016,Supic2019a}.
In self-testing, correlations (e.g. of type $\mathsf{CC} \too \mathsf{CC}$) certify the existence of an underlying valuable quantum resource (say $\mathsf{II} \too \mathsf{QQ}$).
For example, the quantum correlations violating the CHSH inequality~\cite{CHSH} maximally are a signature of an underlying quantum state that is at least as nonclassical as a singlet state (see~\cite{Supic2019a} for a pedagogical derivation).
Recently, the self-testing line of research has expanded beyond self-testing of states, and now has also been applied to steering assemblages~\cite{Supic2016a,Gheorghiu2015}, entangled measurements~\cite{Bancal2018a,Renou2018}, prepare-and-measure scenarios~\cite{Tavakoli2018c}, and quantum gates~\cite{Sekatski2018a}.
However, the correlations that are a signature of the given resource cannot be converted back to that quantum state, and so are not in the same LOSR equivalence class. Rather, they merely allow one to {\em infer the prior existence} of the self-tested resource. As such, the precise relationship with our work is left for exploration.

In the present work, we did not consider the Hilbert space dimensions as part of the resource type.
One could consider a more fine-grained study of conversions between resources of different sizes.
For example, the notion of nonclassical dimension for bipartite quantum states is encoded by the Schmidt rank~\cite{Sperling2011}.
We leave as an open question the generalization of this notion to other resource types; note that Ref.~\cite{rosset2019characterizing} includes a discussion of Hilbert space dimensions solely for the purposes of implementing numerical algorithms.

As a final remark, we recall that the semiquantum games introduced in~\cite{Buscemi2012LOSR} to test bipartite states in a measurement-device independent fashion~\cite{Branciard2013}, can be transformed into guessing games suitable for testing, always in a measurement-device independent fashion, quantum channels and quantum memories~\cite{Buscemi-ProbInfTrans-2016,Rosset-Buscemi-Liang-PRX}. More generally, such single-party guessing games have found application in the context of measurement resources~\cite{Buscemi-ProbInfTrans-2016,Skr-Linden-2019,Skr-Supic-Cavalcanti-2019} and general convex resource theories~\cite{Takagi-etal-2019,Uola-etal-2019,uola2019quantification,Takagi-Regula-PRX}. We leave further investigations about relations between these works and ours for future research.

\section{Conclusions}

We have presented a resource-theoretic framework which unifies various types of resources of nonclassicality which arise when multiple parties have access to classical common causes but no cause-effect relations. This type-independent resource theory allows us to compare the LOSR nonclassicality of resources of arbitrary types and to quantify them using games of arbitrary types. 
We then derived several theorems which ultimately can be used to simplify the methods by which one characterizes the nonclassicality of resources.
Our theorems additionally generalize, unify, and simplify the seminal results of Refs.~\cite{sq,steer,telep}, and our framework leads to a number of exciting questions for future work. 

\section{Acknowledgments}

The authors acknowledge useful discussions with T.C.~Fraser, T.~Gonda, R.W.~Spekkens, and E.~Wolfe and helpful feedback on the manuscript from M.~Hoban, M.~Hall and A.B.~Sainz. D.S. is supported by a Vanier Canada Graduate Scholarship. 
F.B. is grateful for the hospitality of Perimeter Institute where part of this work was carried out and acknowledges partial support from the Japan Society for the Promotion of Science (JSPS) KAKENHI, Grant No.20K03746.
Research at Perimeter Institute is supported in part by the Government of Canada through the Department of Innovation, Science and Economic Development Canada and by the Province of Ontario through the Ministry of Colleges and Universities.
This publication was made possible through the support of a grant from the John Templeton Foundation. The opinions expressed in this publication are those of the authors and do not necessarily reflect the views of the John Templeton Foundation.

\setlength{\bibsep}{2pt plus 1pt minus 2pt}
\bibliographystyle{apsrev4-1}
\nocite{apsrev41Control}
\bibliography{bib}

\begin{thebibliography}{79}%
\makeatletter
\providecommand \@ifxundefined [1]{%
 \@ifx{#1\undefined}
}%
\providecommand \@ifnum [1]{%
 \ifnum #1\expandafter \@firstoftwo
 \else \expandafter \@secondoftwo
 \fi
}%
\providecommand \@ifx [1]{%
 \ifx #1\expandafter \@firstoftwo
 \else \expandafter \@secondoftwo
 \fi
}%
\providecommand \natexlab [1]{#1}%
\providecommand \enquote  [1]{``#1''}%
\providecommand \bibnamefont  [1]{#1}%
\providecommand \bibfnamefont [1]{#1}%
\providecommand \citenamefont [1]{#1}%
\providecommand \href@noop [0]{\@secondoftwo}%
\providecommand \href [0]{\begingroup \@sanitize@url \@href}%
\providecommand \@href[1]{\@@startlink{#1}\@@href}%
\providecommand \@@href[1]{\endgroup#1\@@endlink}%
\providecommand \@sanitize@url [0]{\catcode `\\12\catcode `\$12\catcode
  `\&12\catcode `\#12\catcode `\^12\catcode `\_12\catcode `\%12\relax}%
\providecommand \@@startlink[1]{}%
\providecommand \@@endlink[0]{}%
\providecommand \url  [0]{\begingroup\@sanitize@url \@url }%
\providecommand \@url [1]{\endgroup\@href {#1}{\urlprefix }}%
\providecommand \urlprefix  [0]{URL }%
\providecommand \Eprint [0]{\href }%
\providecommand \doibase [0]{http://dx.doi.org/}%
\providecommand \selectlanguage [0]{\@gobble}%
\providecommand \bibinfo  [0]{\@secondoftwo}%
\providecommand \bibfield  [0]{\@secondoftwo}%
\providecommand \translation [1]{[#1]}%
\providecommand \BibitemOpen [0]{}%
\providecommand \bibitemStop [0]{}%
\providecommand \bibitemNoStop [0]{.\EOS\space}%
\providecommand \EOS [0]{\spacefactor3000\relax}%
\providecommand \BibitemShut  [1]{\csname bibitem#1\endcsname}%
\let\auto@bib@innerbib\@empty
\bibitem [{\citenamefont {Einstein}\ \emph {et~al.}(1935)\citenamefont
  {Einstein}, \citenamefont {Podolsky},\ and\ \citenamefont
  {Rosen}}]{Einstein1935}%
  \BibitemOpen
  \bibfield  {author} {\bibinfo {author} {\bibfnamefont {A.}~\bibnamefont
  {Einstein}}, \bibinfo {author} {\bibfnamefont {B.}~\bibnamefont {Podolsky}},
  \ and\ \bibinfo {author} {\bibfnamefont {N.}~\bibnamefont {Rosen}},\ }\href
  {\doibase 10.1103/PhysRev.47.777} {\bibfield  {journal} {\bibinfo  {journal}
  {Physical Review}\ }\textbf {\bibinfo {volume} {47}},\ \bibinfo {pages} {777}
  (\bibinfo {year} {1935})}\BibitemShut {NoStop}%
\bibitem [{\citenamefont {Wood}\ and\ \citenamefont
  {Spekkens}(2015)}]{Wood2015}%
  \BibitemOpen
  \bibfield  {author} {\bibinfo {author} {\bibfnamefont {C.~J.}\ \bibnamefont
  {Wood}}\ and\ \bibinfo {author} {\bibfnamefont {R.~W.}\ \bibnamefont
  {Spekkens}},\ }\href {http://stacks.iop.org/1367-2630/17/i=3/a=033002}
  {\bibfield  {journal} {\bibinfo  {journal} {New J. Phys.}\ }\textbf {\bibinfo
  {volume} {17}},\ \bibinfo {pages} {033002} (\bibinfo {year}
  {2015})}\BibitemShut {NoStop}%
\bibitem [{\citenamefont {Costa}\ and\ \citenamefont
  {Shrapnel}(2016)}]{Costa_2016}%
  \BibitemOpen
  \bibfield  {author} {\bibinfo {author} {\bibfnamefont {F.}~\bibnamefont
  {Costa}}\ and\ \bibinfo {author} {\bibfnamefont {S.}~\bibnamefont
  {Shrapnel}},\ }\href {\doibase 10.1088/1367-2630/18/6/063032} {\bibfield
  {journal} {\bibinfo  {journal} {New Journal of Physics}\ }\textbf {\bibinfo
  {volume} {18}},\ \bibinfo {pages} {063032} (\bibinfo {year}
  {2016})}\BibitemShut {NoStop}%
\bibitem [{\citenamefont {Allen}\ \emph {et~al.}(2017)\citenamefont {Allen},
  \citenamefont {Barrett}, \citenamefont {Horsman}, \citenamefont {Lee},\ and\
  \citenamefont {Spekkens}}]{Allen2017}%
  \BibitemOpen
  \bibfield  {author} {\bibinfo {author} {\bibfnamefont {J.-M.~A.}\
  \bibnamefont {Allen}}, \bibinfo {author} {\bibfnamefont {J.}~\bibnamefont
  {Barrett}}, \bibinfo {author} {\bibfnamefont {D.~C.}\ \bibnamefont
  {Horsman}}, \bibinfo {author} {\bibfnamefont {C.~M.}\ \bibnamefont {Lee}}, \
  and\ \bibinfo {author} {\bibfnamefont {R.~W.}\ \bibnamefont {Spekkens}},\
  }\href {\doibase 10.1103/PhysRevX.7.031021} {\bibfield  {journal} {\bibinfo
  {journal} {Physical Review X}\ }\textbf {\bibinfo {volume} {7}},\ \bibinfo
  {pages} {031021} (\bibinfo {year} {2017})}\BibitemShut {NoStop}%
\bibitem [{\citenamefont {{Barrett}}\ \emph {et~al.}(2019)\citenamefont
  {{Barrett}}, \citenamefont {{Lorenz}},\ and\ \citenamefont
  {{Oreshkov}}}]{Lorenz2019}%
  \BibitemOpen
  \bibfield  {author} {\bibinfo {author} {\bibfnamefont {J.}~\bibnamefont
  {{Barrett}}}, \bibinfo {author} {\bibfnamefont {R.}~\bibnamefont {{Lorenz}}},
  \ and\ \bibinfo {author} {\bibfnamefont {O.}~\bibnamefont {{Oreshkov}}},\
  }\href@noop {} {\  (\bibinfo {year} {2019})},\ \Eprint
  {http://arxiv.org/abs/1906.10726} {arXiv:1906.10726} \BibitemShut {NoStop}%
\bibitem [{\citenamefont {Horodecki}\ \emph {et~al.}(2009)\citenamefont
  {Horodecki}, \citenamefont {Horodecki}, \citenamefont {Horodecki},\ and\
  \citenamefont {Horodecki}}]{horodecki2009quantum}%
  \BibitemOpen
  \bibfield  {author} {\bibinfo {author} {\bibfnamefont {R.}~\bibnamefont
  {Horodecki}}, \bibinfo {author} {\bibfnamefont {P.}~\bibnamefont
  {Horodecki}}, \bibinfo {author} {\bibfnamefont {M.}~\bibnamefont
  {Horodecki}}, \ and\ \bibinfo {author} {\bibfnamefont {K.}~\bibnamefont
  {Horodecki}},\ }\href@noop {} {\bibfield  {journal} {\bibinfo  {journal}
  {Rev. Mod. Phys.}\ }\textbf {\bibinfo {volume} {81}},\ \bibinfo {pages} {865}
  (\bibinfo {year} {2009})}\BibitemShut {NoStop}%
\bibitem [{\citenamefont {Brunner}\ \emph
  {et~al.}(2014{\natexlab{a}})\citenamefont {Brunner}, \citenamefont
  {Cavalcanti}, \citenamefont {Pironio}, \citenamefont {Scarani},\ and\
  \citenamefont {Wehner}}]{brunner2013Bell}%
  \BibitemOpen
  \bibfield  {author} {\bibinfo {author} {\bibfnamefont {N.}~\bibnamefont
  {Brunner}}, \bibinfo {author} {\bibfnamefont {D.}~\bibnamefont {Cavalcanti}},
  \bibinfo {author} {\bibfnamefont {S.}~\bibnamefont {Pironio}}, \bibinfo
  {author} {\bibfnamefont {V.}~\bibnamefont {Scarani}}, \ and\ \bibinfo
  {author} {\bibfnamefont {S.}~\bibnamefont {Wehner}},\ }\href
  {https://link.aps.org/doi/10.1103/RevModPhys.86.419} {\bibfield  {journal}
  {\bibinfo  {journal} {Rev. Mod. Phys.}\ }\textbf {\bibinfo {volume} {86}},\
  \bibinfo {pages} {419} (\bibinfo {year} {2014}{\natexlab{a}})}\BibitemShut
  {NoStop}%
\bibitem [{\citenamefont {Coecke}\ \emph {et~al.}(2016)\citenamefont {Coecke},
  \citenamefont {Fritz},\ and\ \citenamefont {Spekkens}}]{resthry}%
  \BibitemOpen
  \bibfield  {author} {\bibinfo {author} {\bibfnamefont {B.}~\bibnamefont
  {Coecke}}, \bibinfo {author} {\bibfnamefont {T.}~\bibnamefont {Fritz}}, \
  and\ \bibinfo {author} {\bibfnamefont {R.~W.}\ \bibnamefont {Spekkens}},\
  }\href {\doibase https://doi.org/10.1016/j.ic.2016.02.008} {\bibfield
  {journal} {\bibinfo  {journal} {Information and Computation}\ }\textbf
  {\bibinfo {volume} {250}},\ \bibinfo {pages} {59 } (\bibinfo {year}
  {2016})},\ \bibinfo {note} {{Q}uantum {P}hysics and {L}ogic}\BibitemShut
  {NoStop}%
\bibitem [{\citenamefont {Watrous}(2018)}]{Watrous}%
  \BibitemOpen
  \bibfield  {author} {\bibinfo {author} {\bibfnamefont {J.}~\bibnamefont
  {Watrous}},\ }\href@noop {} {\emph {\bibinfo {title} {The theory of quantum
  information}}}\ (\bibinfo  {publisher} {Cambridge University Press},\
  \bibinfo {year} {2018})\BibitemShut {NoStop}%
\bibitem [{\citenamefont {Wiseman}\ \emph {et~al.}(2007)\citenamefont
  {Wiseman}, \citenamefont {Jones},\ and\ \citenamefont {Doherty}}]{wisesteer}%
  \BibitemOpen
  \bibfield  {author} {\bibinfo {author} {\bibfnamefont {H.~M.}\ \bibnamefont
  {Wiseman}}, \bibinfo {author} {\bibfnamefont {S.~J.}\ \bibnamefont {Jones}},
  \ and\ \bibinfo {author} {\bibfnamefont {A.~C.}\ \bibnamefont {Doherty}},\
  }\href {\doibase 10.1103/PhysRevLett.98.140402} {\bibfield  {journal}
  {\bibinfo  {journal} {Phys. Rev. Lett.}\ }\textbf {\bibinfo {volume} {98}},\
  \bibinfo {pages} {140402} (\bibinfo {year} {2007})}\BibitemShut {NoStop}%
\bibitem [{\citenamefont {Cavalcanti}\ and\ \citenamefont
  {Skrzypczyk}(2017)}]{Cavalcanti2017}%
  \BibitemOpen
  \bibfield  {author} {\bibinfo {author} {\bibfnamefont {D.}~\bibnamefont
  {Cavalcanti}}\ and\ \bibinfo {author} {\bibfnamefont {P.}~\bibnamefont
  {Skrzypczyk}},\ }\href {\doibase 10.1088/1361-6633/80/2/024001} {\bibfield
  {journal} {\bibinfo  {journal} {Reports on Progress in Physics}\ }\textbf
  {\bibinfo {volume} {80}},\ \bibinfo {pages} {024001} (\bibinfo {year}
  {2017})}\BibitemShut {NoStop}%
\bibitem [{\citenamefont {Piani}(2015)}]{channelsteer}%
  \BibitemOpen
  \bibfield  {author} {\bibinfo {author} {\bibfnamefont {M.}~\bibnamefont
  {Piani}},\ }\href {\doibase 10.1364/JOSAB.32.0000A1} {\bibfield  {journal}
  {\bibinfo  {journal} {J. Opt. Soc. Am. B}\ }\textbf {\bibinfo {volume}
  {32}},\ \bibinfo {pages} {A1} (\bibinfo {year} {2015})}\BibitemShut {NoStop}%
\bibitem [{\citenamefont {Hoban}\ and\ \citenamefont
  {Sainz}(2018)}]{Hoban_2018}%
  \BibitemOpen
  \bibfield  {author} {\bibinfo {author} {\bibfnamefont {M.~J.}\ \bibnamefont
  {Hoban}}\ and\ \bibinfo {author} {\bibfnamefont {A.~B.}\ \bibnamefont
  {Sainz}},\ }\href {\doibase 10.1088/1367-2630/aabea8} {\bibfield  {journal}
  {\bibinfo  {journal} {New Journal of Physics}\ }\textbf {\bibinfo {volume}
  {20}},\ \bibinfo {pages} {053048} (\bibinfo {year} {2018})}\BibitemShut
  {NoStop}%
\bibitem [{\citenamefont {Cavalcanti}\ \emph {et~al.}(2017)\citenamefont
  {Cavalcanti}, \citenamefont {Skrzypczyk},\ and\ \citenamefont {\ifmmode
  \check{S}\else \v{S}\fi{}upi\ifmmode~\acute{c}\else \'{c}\fi{}}}]{telep}%
  \BibitemOpen
  \bibfield  {author} {\bibinfo {author} {\bibfnamefont {D.}~\bibnamefont
  {Cavalcanti}}, \bibinfo {author} {\bibfnamefont {P.}~\bibnamefont
  {Skrzypczyk}}, \ and\ \bibinfo {author} {\bibfnamefont {I.}~\bibnamefont
  {\ifmmode \check{S}\else \v{S}\fi{}upi\ifmmode~\acute{c}\else \'{c}\fi{}}},\
  }\href {\doibase 10.1103/PhysRevLett.119.110501} {\bibfield  {journal}
  {\bibinfo  {journal} {Phys. Rev. Lett.}\ }\textbf {\bibinfo {volume} {119}},\
  \bibinfo {pages} {110501} (\bibinfo {year} {2017})}\BibitemShut {NoStop}%
\bibitem [{\citenamefont {Bennett}\ \emph {et~al.}(1999)\citenamefont
  {Bennett}, \citenamefont {DiVincenzo}, \citenamefont {Fuchs}, \citenamefont
  {Mor}, \citenamefont {Rains}, \citenamefont {Shor}, \citenamefont {Smolin},\
  and\ \citenamefont {Wootters}}]{Bennett}%
  \BibitemOpen
  \bibfield  {author} {\bibinfo {author} {\bibfnamefont {C.~H.}\ \bibnamefont
  {Bennett}}, \bibinfo {author} {\bibfnamefont {D.~P.}\ \bibnamefont
  {DiVincenzo}}, \bibinfo {author} {\bibfnamefont {C.~A.}\ \bibnamefont
  {Fuchs}}, \bibinfo {author} {\bibfnamefont {T.}~\bibnamefont {Mor}}, \bibinfo
  {author} {\bibfnamefont {E.}~\bibnamefont {Rains}}, \bibinfo {author}
  {\bibfnamefont {P.~W.}\ \bibnamefont {Shor}}, \bibinfo {author}
  {\bibfnamefont {J.~A.}\ \bibnamefont {Smolin}}, \ and\ \bibinfo {author}
  {\bibfnamefont {W.~K.}\ \bibnamefont {Wootters}},\ }\href {\doibase
  10.1103/PhysRevA.59.1070} {\bibfield  {journal} {\bibinfo  {journal} {Phys.
  Rev. A}\ }\textbf {\bibinfo {volume} {59}},\ \bibinfo {pages} {1070}
  (\bibinfo {year} {1999})}\BibitemShut {NoStop}%
\bibitem [{\citenamefont {Cavalcanti}\ \emph {et~al.}(2013)\citenamefont
  {Cavalcanti}, \citenamefont {Hall},\ and\ \citenamefont {Wiseman}}]{steer}%
  \BibitemOpen
  \bibfield  {author} {\bibinfo {author} {\bibfnamefont {E.~G.}\ \bibnamefont
  {Cavalcanti}}, \bibinfo {author} {\bibfnamefont {M.~J.~W.}\ \bibnamefont
  {Hall}}, \ and\ \bibinfo {author} {\bibfnamefont {H.~M.}\ \bibnamefont
  {Wiseman}},\ }\href {\doibase 10.1103/PhysRevA.87.032306} {\bibfield
  {journal} {\bibinfo  {journal} {Phys. Rev. A}\ }\textbf {\bibinfo {volume}
  {87}},\ \bibinfo {pages} {032306} (\bibinfo {year} {2013})}\BibitemShut
  {NoStop}%
\bibitem [{\citenamefont {{Bel{\'e}n Sainz}}\ \emph {et~al.}(2019)\citenamefont
  {{Bel{\'e}n Sainz}}, \citenamefont {{Hoban}}, \citenamefont {{Skrzypczyk}},\
  and\ \citenamefont {{Aolita}}}]{BobWI}%
  \BibitemOpen
  \bibfield  {author} {\bibinfo {author} {\bibfnamefont {A.}~\bibnamefont
  {{Bel{\'e}n Sainz}}}, \bibinfo {author} {\bibfnamefont {M.~J.}\ \bibnamefont
  {{Hoban}}}, \bibinfo {author} {\bibfnamefont {P.}~\bibnamefont
  {{Skrzypczyk}}}, \ and\ \bibinfo {author} {\bibfnamefont {L.}~\bibnamefont
  {{Aolita}}},\ }\href@noop {} {\  (\bibinfo {year} {2019})},\ \Eprint
  {http://arxiv.org/abs/1907.03705} {arXiv:1907.03705} \BibitemShut {NoStop}%
\bibitem [{\citenamefont {Buscemi}(2012{\natexlab{a}})}]{sq}%
  \BibitemOpen
  \bibfield  {author} {\bibinfo {author} {\bibfnamefont {F.}~\bibnamefont
  {Buscemi}},\ }\href {\doibase 10.1103/PhysRevLett.108.200401} {\bibfield
  {journal} {\bibinfo  {journal} {Phys. Rev. Lett.}\ }\textbf {\bibinfo
  {volume} {108}},\ \bibinfo {pages} {200401} (\bibinfo {year}
  {2012}{\natexlab{a}})}\BibitemShut {NoStop}%
\bibitem [{\citenamefont {Schmid}\ \emph
  {et~al.}(2019{\natexlab{a}})\citenamefont {Schmid}, \citenamefont {Fraser},
  \citenamefont {Kunjwal}, \citenamefont {Sainz}, \citenamefont {Wolfe},\ and\
  \citenamefont {Spekkens}}]{LOSRvsLOCCentang}%
  \BibitemOpen
  \bibfield  {author} {\bibinfo {author} {\bibfnamefont {D.}~\bibnamefont
  {Schmid}}, \bibinfo {author} {\bibfnamefont {T.~C.}\ \bibnamefont {Fraser}},
  \bibinfo {author} {\bibfnamefont {R.}~\bibnamefont {Kunjwal}}, \bibinfo
  {author} {\bibfnamefont {A.~B.}\ \bibnamefont {Sainz}}, \bibinfo {author}
  {\bibfnamefont {E.}~\bibnamefont {Wolfe}}, \ and\ \bibinfo {author}
  {\bibfnamefont {R.~W.}\ \bibnamefont {Spekkens}},\ }\href@noop {} {\enquote
  {\bibinfo {title} {Why standard entanglement theory is inappropriate for the
  study of bell scenarios},}\ } (\bibinfo {year} {2019}{\natexlab{a}}),\
  \Eprint {http://arxiv.org/abs/2004.09194} {arXiv:2004.09194 [quant-ph]}
  \BibitemShut {NoStop}%
\bibitem [{\citenamefont {{de}~Vicente}(2014)}]{de2014nonlocality}%
  \BibitemOpen
  \bibfield  {author} {\bibinfo {author} {\bibfnamefont {J.~I.}\ \bibnamefont
  {{de}~Vicente}},\ }\href {\doibase 10.1088/1751-8113/47/42/424017} {\bibfield
   {journal} {\bibinfo  {journal} {J. Phys. A}\ }\textbf {\bibinfo {volume}
  {47}},\ \bibinfo {pages} {424017} (\bibinfo {year} {2014})}\BibitemShut
  {NoStop}%
\bibitem [{\citenamefont {Gallego}\ and\ \citenamefont
  {Aolita}(2017)}]{gallego2016nonlocality}%
  \BibitemOpen
  \bibfield  {author} {\bibinfo {author} {\bibfnamefont {R.}~\bibnamefont
  {Gallego}}\ and\ \bibinfo {author} {\bibfnamefont {L.}~\bibnamefont
  {Aolita}},\ }\href {https://link.aps.org/doi/10.1103/PhysRevA.95.032118}
  {\bibfield  {journal} {\bibinfo  {journal} {Phys. Rev. A}\ }\textbf {\bibinfo
  {volume} {95}} (\bibinfo {year} {2017})}\BibitemShut {NoStop}%
\bibitem [{\citenamefont {Wolfe}\ \emph {et~al.}(2019)\citenamefont {Wolfe},
  \citenamefont {Schmid}, \citenamefont {Sainz}, \citenamefont {Kunjwal},\ and\
  \citenamefont {Spekkens}}]{Bellquantified}%
  \BibitemOpen
  \bibfield  {author} {\bibinfo {author} {\bibfnamefont {E.}~\bibnamefont
  {Wolfe}}, \bibinfo {author} {\bibfnamefont {D.}~\bibnamefont {Schmid}},
  \bibinfo {author} {\bibfnamefont {A.~B.}\ \bibnamefont {Sainz}}, \bibinfo
  {author} {\bibfnamefont {R.}~\bibnamefont {Kunjwal}}, \ and\ \bibinfo
  {author} {\bibfnamefont {R.~W.}\ \bibnamefont {Spekkens}},\ }\href@noop {}
  {\enquote {\bibinfo {title} {Quantifying {{Bell}}: the resource theory of
  nonclassicality of common-cause boxes},}\ } (\bibinfo {year} {2019}),\
  \Eprint {http://arxiv.org/abs/1903.06311} {arXiv:1903.06311 [quant-ph]}
  \BibitemShut {NoStop}%
\bibitem [{\citenamefont {Bell}(2004)}]{Bell2004}%
  \BibitemOpen
  \bibfield  {author} {\bibinfo {author} {\bibfnamefont {J.~S.}\ \bibnamefont
  {Bell}},\ }\href@noop {} {\emph {\bibinfo {title} {Speakable and
  {{Unspeakable}} in {{Quantum Mechanics}}: {{Collected Papers}} on {{Quantum
  Philosophy}}}}}\ (\bibinfo  {publisher} {{Cambridge University Press}},\
  \bibinfo {year} {2004})\BibitemShut {NoStop}%
\bibitem [{\citenamefont {P{\"u}tz}\ \emph {et~al.}(2014)\citenamefont
  {P{\"u}tz}, \citenamefont {Rosset}, \citenamefont {Barnea}, \citenamefont
  {Liang},\ and\ \citenamefont {Gisin}}]{Putz2014}%
  \BibitemOpen
  \bibfield  {author} {\bibinfo {author} {\bibfnamefont {G.}~\bibnamefont
  {P{\"u}tz}}, \bibinfo {author} {\bibfnamefont {D.}~\bibnamefont {Rosset}},
  \bibinfo {author} {\bibfnamefont {T.~J.}\ \bibnamefont {Barnea}}, \bibinfo
  {author} {\bibfnamefont {Y.-C.}\ \bibnamefont {Liang}}, \ and\ \bibinfo
  {author} {\bibfnamefont {N.}~\bibnamefont {Gisin}},\ }\href {\doibase
  10.1103/PhysRevLett.113.190402} {\bibfield  {journal} {\bibinfo  {journal}
  {Physical Review Letters}\ }\textbf {\bibinfo {volume} {113}},\ \bibinfo
  {pages} {190402} (\bibinfo {year} {2014})}\BibitemShut {NoStop}%
\bibitem [{\citenamefont {Rosset}\ \emph {et~al.}(2012)\citenamefont {Rosset},
  \citenamefont {{Ferretti-Sch{\"o}bitz}}, \citenamefont {Bancal},
  \citenamefont {Gisin},\ and\ \citenamefont {Liang}}]{Rosset2012a}%
  \BibitemOpen
  \bibfield  {author} {\bibinfo {author} {\bibfnamefont {D.}~\bibnamefont
  {Rosset}}, \bibinfo {author} {\bibfnamefont {R.}~\bibnamefont
  {{Ferretti-Sch{\"o}bitz}}}, \bibinfo {author} {\bibfnamefont {J.-D.}\
  \bibnamefont {Bancal}}, \bibinfo {author} {\bibfnamefont {N.}~\bibnamefont
  {Gisin}}, \ and\ \bibinfo {author} {\bibfnamefont {Y.-C.}\ \bibnamefont
  {Liang}},\ }\href {\doibase 10.1103/PhysRevA.86.062325} {\bibfield  {journal}
  {\bibinfo  {journal} {Physical Review A}\ }\textbf {\bibinfo {volume} {86}},\
  \bibinfo {pages} {062325} (\bibinfo {year} {2012})}\BibitemShut {NoStop}%
\bibitem [{\citenamefont {Pironio}\ \emph {et~al.}(2016)\citenamefont
  {Pironio}, \citenamefont {Scarani},\ and\ \citenamefont
  {Vidick}}]{Pironio2016}%
  \BibitemOpen
  \bibfield  {author} {\bibinfo {author} {\bibfnamefont {S.}~\bibnamefont
  {Pironio}}, \bibinfo {author} {\bibfnamefont {V.}~\bibnamefont {Scarani}}, \
  and\ \bibinfo {author} {\bibfnamefont {T.}~\bibnamefont {Vidick}},\ }\href
  {\doibase 10.1088/1367-2630/18/10/100202} {\bibfield  {journal} {\bibinfo
  {journal} {New Journal of Physics}\ }\textbf {\bibinfo {volume} {18}},\
  \bibinfo {pages} {100202} (\bibinfo {year} {2016})}\BibitemShut {NoStop}%
\bibitem [{\citenamefont {Cavalcanti}\ \emph {et~al.}(2015)\citenamefont
  {Cavalcanti}, \citenamefont {Skrzypczyk}, \citenamefont {Aguilar},
  \citenamefont {Nery}, \citenamefont {Ribeiro},\ and\ \citenamefont
  {Walborn}}]{Cavalcanti2015}%
  \BibitemOpen
  \bibfield  {author} {\bibinfo {author} {\bibfnamefont {D.}~\bibnamefont
  {Cavalcanti}}, \bibinfo {author} {\bibfnamefont {P.}~\bibnamefont
  {Skrzypczyk}}, \bibinfo {author} {\bibfnamefont {G.~H.}\ \bibnamefont
  {Aguilar}}, \bibinfo {author} {\bibfnamefont {R.~V.}\ \bibnamefont {Nery}},
  \bibinfo {author} {\bibfnamefont {P.~H.~S.}\ \bibnamefont {Ribeiro}}, \ and\
  \bibinfo {author} {\bibfnamefont {S.~P.}\ \bibnamefont {Walborn}},\ }\href
  {\doibase 10.1038/ncomms8941} {\bibfield  {journal} {\bibinfo  {journal}
  {Nature Communications}\ }\textbf {\bibinfo {volume} {6}},\ \bibinfo {pages}
  {7941} (\bibinfo {year} {2015})}\BibitemShut {NoStop}%
\bibitem [{\citenamefont {Werner}(1989)}]{Werner}%
  \BibitemOpen
  \bibfield  {author} {\bibinfo {author} {\bibfnamefont {R.~F.}\ \bibnamefont
  {Werner}},\ }\href {\doibase 10.1103/PhysRevA.40.4277} {\bibfield  {journal}
  {\bibinfo  {journal} {Phys. Rev. A}\ }\textbf {\bibinfo {volume} {40}},\
  \bibinfo {pages} {4277} (\bibinfo {year} {1989})}\BibitemShut {NoStop}%
\bibitem [{\citenamefont {Barrett}(2002)}]{Barrettlocal}%
  \BibitemOpen
  \bibfield  {author} {\bibinfo {author} {\bibfnamefont {J.}~\bibnamefont
  {Barrett}},\ }\href {\doibase 10.1103/PhysRevA.65.042302} {\bibfield
  {journal} {\bibinfo  {journal} {Phys. Rev. A}\ }\textbf {\bibinfo {volume}
  {65}},\ \bibinfo {pages} {042302} (\bibinfo {year} {2002})}\BibitemShut
  {NoStop}%
\bibitem [{\citenamefont {{Lipka-Bartosik}}\ and\ \citenamefont
  {Skrzypczyk}(2019)}]{LipkaBartosik2019}%
  \BibitemOpen
  \bibfield  {author} {\bibinfo {author} {\bibfnamefont {P.}~\bibnamefont
  {{Lipka-Bartosik}}}\ and\ \bibinfo {author} {\bibfnamefont {P.}~\bibnamefont
  {Skrzypczyk}},\ }\href@noop {} {\  (\bibinfo {year} {2019})},\ \Eprint
  {http://arxiv.org/abs/1908.05107} {arXiv:1908.05107 [quant-ph]} \BibitemShut
  {NoStop}%
\bibitem [{\citenamefont {Branciard}\ \emph {et~al.}(2013)\citenamefont
  {Branciard}, \citenamefont {Rosset}, \citenamefont {Liang},\ and\
  \citenamefont {Gisin}}]{Branciard2013}%
  \BibitemOpen
  \bibfield  {author} {\bibinfo {author} {\bibfnamefont {C.}~\bibnamefont
  {Branciard}}, \bibinfo {author} {\bibfnamefont {D.}~\bibnamefont {Rosset}},
  \bibinfo {author} {\bibfnamefont {Y.-C.}\ \bibnamefont {Liang}}, \ and\
  \bibinfo {author} {\bibfnamefont {N.}~\bibnamefont {Gisin}},\ }\href
  {\doibase 10.1103/PhysRevLett.110.060405} {\bibfield  {journal} {\bibinfo
  {journal} {Physical Review Letters}\ }\textbf {\bibinfo {volume} {110}},\
  \bibinfo {pages} {060405} (\bibinfo {year} {2013})}\BibitemShut {NoStop}%
\bibitem [{\citenamefont {Rosset}\ \emph {et~al.}(2013)\citenamefont {Rosset},
  \citenamefont {Branciard}, \citenamefont {Gisin},\ and\ \citenamefont
  {Liang}}]{Rosset2013a}%
  \BibitemOpen
  \bibfield  {author} {\bibinfo {author} {\bibfnamefont {D.}~\bibnamefont
  {Rosset}}, \bibinfo {author} {\bibfnamefont {C.}~\bibnamefont {Branciard}},
  \bibinfo {author} {\bibfnamefont {N.}~\bibnamefont {Gisin}}, \ and\ \bibinfo
  {author} {\bibfnamefont {Y.-C.}\ \bibnamefont {Liang}},\ }\href {\doibase
  10.1088/1367-2630/15/5/053025} {\bibfield  {journal} {\bibinfo  {journal}
  {New Journal of Physics}\ }\textbf {\bibinfo {volume} {15}},\ \bibinfo
  {pages} {053025} (\bibinfo {year} {2013})}\BibitemShut {NoStop}%
\bibitem [{\citenamefont {Shahandeh}\ \emph {et~al.}(2017)\citenamefont
  {Shahandeh}, \citenamefont {Hall},\ and\ \citenamefont
  {Ralph}}]{Shahandeh2017}%
  \BibitemOpen
  \bibfield  {author} {\bibinfo {author} {\bibfnamefont {F.}~\bibnamefont
  {Shahandeh}}, \bibinfo {author} {\bibfnamefont {M.~J.~W.}\ \bibnamefont
  {Hall}}, \ and\ \bibinfo {author} {\bibfnamefont {T.~C.}\ \bibnamefont
  {Ralph}},\ }\href {\doibase 10.1103/PhysRevLett.118.150505} {\bibfield
  {journal} {\bibinfo  {journal} {Physical Review Letters}\ }\textbf {\bibinfo
  {volume} {118}},\ \bibinfo {pages} {150505} (\bibinfo {year}
  {2017})}\BibitemShut {NoStop}%
\bibitem [{\citenamefont {Rosset}\ \emph
  {et~al.}(2018{\natexlab{a}})\citenamefont {Rosset}, \citenamefont {Martin},
  \citenamefont {Verbanis}, \citenamefont {Lim},\ and\ \citenamefont
  {Thew}}]{Rosset2018a}%
  \BibitemOpen
  \bibfield  {author} {\bibinfo {author} {\bibfnamefont {D.}~\bibnamefont
  {Rosset}}, \bibinfo {author} {\bibfnamefont {A.}~\bibnamefont {Martin}},
  \bibinfo {author} {\bibfnamefont {E.}~\bibnamefont {Verbanis}}, \bibinfo
  {author} {\bibfnamefont {C.~C.~W.}\ \bibnamefont {Lim}}, \ and\ \bibinfo
  {author} {\bibfnamefont {R.}~\bibnamefont {Thew}},\ }\href {\doibase
  10.1103/PhysRevA.98.052332} {\bibfield  {journal} {\bibinfo  {journal}
  {Physical Review A}\ }\textbf {\bibinfo {volume} {98}},\ \bibinfo {pages}
  {052332} (\bibinfo {year} {2018}{\natexlab{a}})}\BibitemShut {NoStop}%
\bibitem [{\citenamefont {Vidal}(2000)}]{vidal2000entanglement}%
  \BibitemOpen
  \bibfield  {author} {\bibinfo {author} {\bibfnamefont {G.}~\bibnamefont
  {Vidal}},\ }\href {\doibase 10.1080/09500340008244048} {\bibfield  {journal}
  {\bibinfo  {journal} {J. Mod. Optic.}\ }\textbf {\bibinfo {volume} {47}},\
  \bibinfo {pages} {355} (\bibinfo {year} {2000})}\BibitemShut {NoStop}%
\bibitem [{\citenamefont {Rosset}\ \emph {et~al.}(2019)\citenamefont {Rosset},
  \citenamefont {Schmid},\ and\ \citenamefont
  {Buscemi}}]{rosset2019characterizing}%
  \BibitemOpen
  \bibfield  {author} {\bibinfo {author} {\bibfnamefont {D.}~\bibnamefont
  {Rosset}}, \bibinfo {author} {\bibfnamefont {D.}~\bibnamefont {Schmid}}, \
  and\ \bibinfo {author} {\bibfnamefont {F.}~\bibnamefont {Buscemi}},\
  }\href@noop {} {\enquote {\bibinfo {title} {Characterizing nonclassicality of
  arbitrary distributed devices},}\ } (\bibinfo {year} {2019}),\ \Eprint
  {http://arxiv.org/abs/1911.12462} {arXiv:1911.12462 [quant-ph]} \BibitemShut
  {NoStop}%
\bibitem [{\citenamefont {Brunner}\ \emph
  {et~al.}(2014{\natexlab{b}})\citenamefont {Brunner}, \citenamefont
  {Cavalcanti}, \citenamefont {Pironio}, \citenamefont {Scarani},\ and\
  \citenamefont {Wehner}}]{Bellreview}%
  \BibitemOpen
  \bibfield  {author} {\bibinfo {author} {\bibfnamefont {N.}~\bibnamefont
  {Brunner}}, \bibinfo {author} {\bibfnamefont {D.}~\bibnamefont {Cavalcanti}},
  \bibinfo {author} {\bibfnamefont {S.}~\bibnamefont {Pironio}}, \bibinfo
  {author} {\bibfnamefont {V.}~\bibnamefont {Scarani}}, \ and\ \bibinfo
  {author} {\bibfnamefont {S.}~\bibnamefont {Wehner}},\ }\href
  {https://link.aps.org/doi/10.1103/RevModPhys.86.419} {\bibfield  {journal}
  {\bibinfo  {journal} {Rev. Mod. Phys.}\ }\textbf {\bibinfo {volume} {86}},\
  \bibinfo {pages} {419} (\bibinfo {year} {2014}{\natexlab{b}})}\BibitemShut
  {NoStop}%
\bibitem [{\citenamefont {Chru{\'s}ci{\'n}ski}\ and\ \citenamefont
  {Sarbicki}(2014)}]{Chruscinski2014}%
  \BibitemOpen
  \bibfield  {author} {\bibinfo {author} {\bibfnamefont {D.}~\bibnamefont
  {Chru{\'s}ci{\'n}ski}}\ and\ \bibinfo {author} {\bibfnamefont
  {G.}~\bibnamefont {Sarbicki}},\ }\href {\doibase
  10.1088/1751-8113/47/48/483001} {\bibfield  {journal} {\bibinfo  {journal}
  {Journal of Physics A: Mathematical and Theoretical}\ }\textbf {\bibinfo
  {volume} {47}},\ \bibinfo {pages} {483001} (\bibinfo {year}
  {2014})}\BibitemShut {NoStop}%
\bibitem [{\citenamefont {Popescu}\ and\ \citenamefont
  {Rohrlich}(1994)}]{Popescu1994}%
  \BibitemOpen
  \bibfield  {author} {\bibinfo {author} {\bibfnamefont {S.}~\bibnamefont
  {Popescu}}\ and\ \bibinfo {author} {\bibfnamefont {D.}~\bibnamefont
  {Rohrlich}},\ }\href {\doibase 10.1007/BF02058098} {\bibfield  {journal}
  {\bibinfo  {journal} {Foundations of Physics}\ }\textbf {\bibinfo {volume}
  {24}},\ \bibinfo {pages} {379} (\bibinfo {year} {1994})}\BibitemShut
  {NoStop}%
\bibitem [{\citenamefont {Barrett}\ \emph {et~al.}(2005)\citenamefont
  {Barrett}, \citenamefont {Linden}, \citenamefont {Massar}, \citenamefont
  {Pironio}, \citenamefont {Popescu},\ and\ \citenamefont
  {Roberts}}]{Barrett2005}%
  \BibitemOpen
  \bibfield  {author} {\bibinfo {author} {\bibfnamefont {J.}~\bibnamefont
  {Barrett}}, \bibinfo {author} {\bibfnamefont {N.}~\bibnamefont {Linden}},
  \bibinfo {author} {\bibfnamefont {S.}~\bibnamefont {Massar}}, \bibinfo
  {author} {\bibfnamefont {S.}~\bibnamefont {Pironio}}, \bibinfo {author}
  {\bibfnamefont {S.}~\bibnamefont {Popescu}}, \ and\ \bibinfo {author}
  {\bibfnamefont {D.}~\bibnamefont {Roberts}},\ }\href {\doibase
  10.1103/PhysRevA.71.022101} {\bibfield  {journal} {\bibinfo  {journal}
  {Physical Review A}\ }\textbf {\bibinfo {volume} {71}},\ \bibinfo {pages}
  {022101} (\bibinfo {year} {2005})}\BibitemShut {NoStop}%
\bibitem [{\citenamefont {Beckman}\ \emph {et~al.}(2001)\citenamefont
  {Beckman}, \citenamefont {Gottesman}, \citenamefont {Nielsen},\ and\
  \citenamefont {Preskill}}]{causallocaliz}%
  \BibitemOpen
  \bibfield  {author} {\bibinfo {author} {\bibfnamefont {D.}~\bibnamefont
  {Beckman}}, \bibinfo {author} {\bibfnamefont {D.}~\bibnamefont {Gottesman}},
  \bibinfo {author} {\bibfnamefont {M.~A.}\ \bibnamefont {Nielsen}}, \ and\
  \bibinfo {author} {\bibfnamefont {J.}~\bibnamefont {Preskill}},\ }\href
  {\doibase 10.1103/PhysRevA.64.052309} {\bibfield  {journal} {\bibinfo
  {journal} {Phys. Rev. A}\ }\textbf {\bibinfo {volume} {64}},\ \bibinfo
  {pages} {052309} (\bibinfo {year} {2001})}\BibitemShut {NoStop}%
\bibitem [{\citenamefont {Nielsen}\ and\ \citenamefont
  {Chuang}(2010)}]{NielsenAndChuang}%
  \BibitemOpen
  \bibfield  {author} {\bibinfo {author} {\bibfnamefont {M.~A.}\ \bibnamefont
  {Nielsen}}\ and\ \bibinfo {author} {\bibfnamefont {I.~L.}\ \bibnamefont
  {Chuang}},\ }\href {\doibase 10.1017/CBO9780511976667} {\emph {\bibinfo
  {title} {{Quantum Computation and Quantum Information}}}}\ (\bibinfo
  {publisher} {Cambridge University Press},\ \bibinfo {year}
  {2010})\BibitemShut {NoStop}%
\bibitem [{\citenamefont {Schmid}\ \emph
  {et~al.}(2019{\natexlab{b}})\citenamefont {Schmid}, \citenamefont {Ried},\
  and\ \citenamefont {Spekkens}}]{Schmidcausal}%
  \BibitemOpen
  \bibfield  {author} {\bibinfo {author} {\bibfnamefont {D.}~\bibnamefont
  {Schmid}}, \bibinfo {author} {\bibfnamefont {K.}~\bibnamefont {Ried}}, \ and\
  \bibinfo {author} {\bibfnamefont {R.~W.}\ \bibnamefont {Spekkens}},\ }\href
  {\doibase 10.1103/PhysRevA.100.022112} {\bibfield  {journal} {\bibinfo
  {journal} {Phys. Rev. A}\ }\textbf {\bibinfo {volume} {100}},\ \bibinfo
  {pages} {022112} (\bibinfo {year} {2019}{\natexlab{b}})}\BibitemShut
  {NoStop}%
\bibitem [{\citenamefont {Horodecki}\ \emph {et~al.}(2015)\citenamefont
  {Horodecki}, \citenamefont {Grudka}, \citenamefont {Joshi}, \citenamefont
  {K{\l}obus},\ and\ \citenamefont {{\L}odyga}}]{Horodecki2015}%
  \BibitemOpen
  \bibfield  {author} {\bibinfo {author} {\bibfnamefont {K.}~\bibnamefont
  {Horodecki}}, \bibinfo {author} {\bibfnamefont {A.}~\bibnamefont {Grudka}},
  \bibinfo {author} {\bibfnamefont {P.}~\bibnamefont {Joshi}}, \bibinfo
  {author} {\bibfnamefont {W.}~\bibnamefont {K{\l}obus}}, \ and\ \bibinfo
  {author} {\bibfnamefont {J.}~\bibnamefont {{\L}odyga}},\ }\href {\doibase
  10.1103/PhysRevA.92.032104} {\bibfield  {journal} {\bibinfo  {journal}
  {Physical Review A}\ }\textbf {\bibinfo {volume} {92}},\ \bibinfo {pages}
  {032104} (\bibinfo {year} {2015})}\BibitemShut {NoStop}%
\bibitem [{\citenamefont {de~Vicente}(2014)}]{Vicente2014}%
  \BibitemOpen
  \bibfield  {author} {\bibinfo {author} {\bibfnamefont {J.~I.}\ \bibnamefont
  {de~Vicente}},\ }\href {\doibase 10.1088/1751-8113/47/42/424017} {\bibfield
  {journal} {\bibinfo  {journal} {Journal of Physics A: Mathematical and
  Theoretical}\ }\textbf {\bibinfo {volume} {47}},\ \bibinfo {pages} {424017}
  (\bibinfo {year} {2014})}\BibitemShut {NoStop}%
\bibitem [{\citenamefont {Schrodinger}(1935)}]{schrodinger_1935}%
  \BibitemOpen
  \bibfield  {author} {\bibinfo {author} {\bibfnamefont {E.}~\bibnamefont
  {Schrodinger}},\ }\href {\doibase 10.1017/S0305004100013554} {\bibfield
  {journal} {\bibinfo  {journal} {Mathematical Proceedings of the Cambridge
  Philosophical Society}\ }\textbf {\bibinfo {volume} {31}},\ \bibinfo {pages}
  {555} (\bibinfo {year} {1935})}\BibitemShut {NoStop}%
\bibitem [{\citenamefont {Skrzypczyk}\ \emph {et~al.}(2014)\citenamefont
  {Skrzypczyk}, \citenamefont {Navascu{\'e}s},\ and\ \citenamefont
  {Cavalcanti}}]{Skrzypczyk2014}%
  \BibitemOpen
  \bibfield  {author} {\bibinfo {author} {\bibfnamefont {P.}~\bibnamefont
  {Skrzypczyk}}, \bibinfo {author} {\bibfnamefont {M.}~\bibnamefont
  {Navascu{\'e}s}}, \ and\ \bibinfo {author} {\bibfnamefont {D.}~\bibnamefont
  {Cavalcanti}},\ }\href {\doibase 10.1103/PhysRevLett.112.180404} {\bibfield
  {journal} {\bibinfo  {journal} {Physical Review Letters}\ }\textbf {\bibinfo
  {volume} {112}},\ \bibinfo {pages} {180404} (\bibinfo {year}
  {2014})}\BibitemShut {NoStop}%
\bibitem [{\citenamefont {Gallego}\ and\ \citenamefont
  {Aolita}(2015)}]{Gallego2015}%
  \BibitemOpen
  \bibfield  {author} {\bibinfo {author} {\bibfnamefont {R.}~\bibnamefont
  {Gallego}}\ and\ \bibinfo {author} {\bibfnamefont {L.}~\bibnamefont
  {Aolita}},\ }\href {\doibase 10.1103/PhysRevX.5.041008} {\bibfield  {journal}
  {\bibinfo  {journal} {Physical Review X}\ }\textbf {\bibinfo {volume} {5}},\
  \bibinfo {pages} {041008} (\bibinfo {year} {2015})}\BibitemShut {NoStop}%
\bibitem [{\citenamefont {Piani}\ and\ \citenamefont
  {Watrous}(2015)}]{Piani2015a}%
  \BibitemOpen
  \bibfield  {author} {\bibinfo {author} {\bibfnamefont {M.}~\bibnamefont
  {Piani}}\ and\ \bibinfo {author} {\bibfnamefont {J.}~\bibnamefont
  {Watrous}},\ }\href {\doibase 10.1103/PhysRevLett.114.060404} {\bibfield
  {journal} {\bibinfo  {journal} {Physical Review Letters}\ }\textbf {\bibinfo
  {volume} {114}},\ \bibinfo {pages} {060404} (\bibinfo {year}
  {2015})}\BibitemShut {NoStop}%
\bibitem [{\citenamefont {Uola}\ \emph
  {et~al.}(2019{\natexlab{a}})\citenamefont {Uola}, \citenamefont {Costa},
  \citenamefont {Nguyen},\ and\ \citenamefont {G{\"u}hne}}]{Uola2019}%
  \BibitemOpen
  \bibfield  {author} {\bibinfo {author} {\bibfnamefont {R.}~\bibnamefont
  {Uola}}, \bibinfo {author} {\bibfnamefont {A.~C.~S.}\ \bibnamefont {Costa}},
  \bibinfo {author} {\bibfnamefont {H.~C.}\ \bibnamefont {Nguyen}}, \ and\
  \bibinfo {author} {\bibfnamefont {O.}~\bibnamefont {G{\"u}hne}},\ }\href@noop
  {} {\  (\bibinfo {year} {2019}{\natexlab{a}})},\ \Eprint
  {http://arxiv.org/abs/1903.06663} {arXiv:1903.06663} \BibitemShut {NoStop}%
\bibitem [{\citenamefont {Pusey}(2013)}]{Pusey2013}%
  \BibitemOpen
  \bibfield  {author} {\bibinfo {author} {\bibfnamefont {M.~F.}\ \bibnamefont
  {Pusey}},\ }\href {\doibase 10.1103/PhysRevA.88.032313} {\bibfield  {journal}
  {\bibinfo  {journal} {Phys. Rev. A}\ }\textbf {\bibinfo {volume} {88}},\
  \bibinfo {pages} {032313} (\bibinfo {year} {2013})}\BibitemShut {NoStop}%
\bibitem [{\citenamefont {{\v S}upi{\'c}}\ \emph {et~al.}(2018)\citenamefont
  {{\v S}upi{\'c}}, \citenamefont {Skrzypczyk},\ and\ \citenamefont
  {Cavalcanti}}]{Supic2018}%
  \BibitemOpen
  \bibfield  {author} {\bibinfo {author} {\bibfnamefont {I.}~\bibnamefont {{\v
  S}upi{\'c}}}, \bibinfo {author} {\bibfnamefont {P.}~\bibnamefont
  {Skrzypczyk}}, \ and\ \bibinfo {author} {\bibfnamefont {D.}~\bibnamefont
  {Cavalcanti}},\ }\href@noop {} {\  (\bibinfo {year} {2018})},\ \Eprint
  {http://arxiv.org/abs/1804.10612} {arXiv:1804.10612} \BibitemShut {NoStop}%
\bibitem [{\citenamefont {Chiribella}\ \emph {et~al.}(2009)\citenamefont
  {Chiribella}, \citenamefont {D'Ariano},\ and\ \citenamefont
  {Perinotti}}]{qcombs09}%
  \BibitemOpen
  \bibfield  {author} {\bibinfo {author} {\bibfnamefont {G.}~\bibnamefont
  {Chiribella}}, \bibinfo {author} {\bibfnamefont {G.~M.}\ \bibnamefont
  {D'Ariano}}, \ and\ \bibinfo {author} {\bibfnamefont {P.}~\bibnamefont
  {Perinotti}},\ }\href {https://link.aps.org/doi/10.1103/PhysRevA.80.022339}
  {\bibfield  {journal} {\bibinfo  {journal} {Phys. Rev. A}\ }\textbf {\bibinfo
  {volume} {80}},\ \bibinfo {pages} {022339} (\bibinfo {year}
  {2009})}\BibitemShut {NoStop}%
\bibitem [{\citenamefont {Bennett}\ \emph {et~al.}(1993)\citenamefont
  {Bennett}, \citenamefont {Brassard}, \citenamefont {Cr\'epeau}, \citenamefont
  {Jozsa}, \citenamefont {Peres},\ and\ \citenamefont {Wootters}}]{Bennett93}%
  \BibitemOpen
  \bibfield  {author} {\bibinfo {author} {\bibfnamefont {C.~H.}\ \bibnamefont
  {Bennett}}, \bibinfo {author} {\bibfnamefont {G.}~\bibnamefont {Brassard}},
  \bibinfo {author} {\bibfnamefont {C.}~\bibnamefont {Cr\'epeau}}, \bibinfo
  {author} {\bibfnamefont {R.}~\bibnamefont {Jozsa}}, \bibinfo {author}
  {\bibfnamefont {A.}~\bibnamefont {Peres}}, \ and\ \bibinfo {author}
  {\bibfnamefont {W.~K.}\ \bibnamefont {Wootters}},\ }\href {\doibase
  10.1103/PhysRevLett.70.1895} {\bibfield  {journal} {\bibinfo  {journal}
  {Phys. Rev. Lett.}\ }\textbf {\bibinfo {volume} {70}},\ \bibinfo {pages}
  {1895} (\bibinfo {year} {1993})}\BibitemShut {NoStop}%
\bibitem [{\citenamefont {Lipka-Bartosik}\ and\ \citenamefont
  {Skrzypczyk}(2019)}]{lipkabartosik2019operational}%
  \BibitemOpen
  \bibfield  {author} {\bibinfo {author} {\bibfnamefont {P.}~\bibnamefont
  {Lipka-Bartosik}}\ and\ \bibinfo {author} {\bibfnamefont {P.}~\bibnamefont
  {Skrzypczyk}},\ }\href@noop {} {\enquote {\bibinfo {title} {The operational
  advantages provided by non-classical teleportation},}\ } (\bibinfo {year}
  {2019}),\ \Eprint {http://arxiv.org/abs/1908.05107} {arXiv:1908.05107
  [quant-ph]} \BibitemShut {NoStop}%
\bibitem [{\citenamefont {{\v S}upi{\'c}}\ \emph {et~al.}(2017)\citenamefont
  {{\v S}upi{\'c}}, \citenamefont {Skrzypczyk},\ and\ \citenamefont
  {Cavalcanti}}]{Supic2017}%
  \BibitemOpen
  \bibfield  {author} {\bibinfo {author} {\bibfnamefont {I.}~\bibnamefont {{\v
  S}upi{\'c}}}, \bibinfo {author} {\bibfnamefont {P.}~\bibnamefont
  {Skrzypczyk}}, \ and\ \bibinfo {author} {\bibfnamefont {D.}~\bibnamefont
  {Cavalcanti}},\ }\href {\doibase 10.1103/PhysRevA.95.042340} {\bibfield
  {journal} {\bibinfo  {journal} {Physical Review A}\ }\textbf {\bibinfo
  {volume} {95}},\ \bibinfo {pages} {042340} (\bibinfo {year}
  {2017})}\BibitemShut {NoStop}%
\bibitem [{\citenamefont {{Gonda}}\ and\ \citenamefont
  {{Spekkens}}(2019)}]{Gonda2019}%
  \BibitemOpen
  \bibfield  {author} {\bibinfo {author} {\bibfnamefont {T.}~\bibnamefont
  {{Gonda}}}\ and\ \bibinfo {author} {\bibfnamefont {R.~W.}\ \bibnamefont
  {{Spekkens}}},\ }\href@noop {} {\  (\bibinfo {year} {2019})},\ \Eprint
  {http://arxiv.org/abs/1912.07085} {arXiv:1912.07085} \BibitemShut {NoStop}%
\bibitem [{\citenamefont {D\"ur}\ \emph {et~al.}(2000)\citenamefont {D\"ur},
  \citenamefont {Vidal},\ and\ \citenamefont {Cirac}}]{Dur2000twoways}%
  \BibitemOpen
  \bibfield  {author} {\bibinfo {author} {\bibfnamefont {W.}~\bibnamefont
  {D\"ur}}, \bibinfo {author} {\bibfnamefont {G.}~\bibnamefont {Vidal}}, \ and\
  \bibinfo {author} {\bibfnamefont {J.~I.}\ \bibnamefont {Cirac}},\ }\href
  {https://link.aps.org/doi/10.1103/PhysRevA.62.062314} {\bibfield  {journal}
  {\bibinfo  {journal} {Phys. Rev. A}\ }\textbf {\bibinfo {volume} {62}},\
  \bibinfo {pages} {062314} (\bibinfo {year} {2000})}\BibitemShut {NoStop}%
\bibitem [{\citenamefont {{Schmid}}\ \emph {et~al.}(2020)\citenamefont
  {{Schmid}}, \citenamefont {{Du}}, \citenamefont {{Mudassar}}, \citenamefont
  {{Coulter-de Wit}}, \citenamefont {{Rosset}},\ and\ \citenamefont
  {{Hoban}}}]{postquantumschmid}%
  \BibitemOpen
  \bibfield  {author} {\bibinfo {author} {\bibfnamefont {D.}~\bibnamefont
  {{Schmid}}}, \bibinfo {author} {\bibfnamefont {H.}~\bibnamefont {{Du}}},
  \bibinfo {author} {\bibfnamefont {M.}~\bibnamefont {{Mudassar}}}, \bibinfo
  {author} {\bibfnamefont {G.}~\bibnamefont {{Coulter-de Wit}}}, \bibinfo
  {author} {\bibfnamefont {D.}~\bibnamefont {{Rosset}}}, \ and\ \bibinfo
  {author} {\bibfnamefont {M.~J.}\ \bibnamefont {{Hoban}}},\ }\href@noop {}
  {\enquote {\bibinfo {title} {Postquantum common-cause channels: the resource
  theory of local operations and shared entanglement},}\ } (\bibinfo {year}
  {2020}),\ \Eprint {http://arxiv.org/abs/2004.06133} {arXiv:2004.06133
  [quant-ph]} \BibitemShut {NoStop}%
\bibitem [{\citenamefont {Mayers}\ and\ \citenamefont
  {Yao}(2004)}]{Mayers2004}%
  \BibitemOpen
  \bibfield  {author} {\bibinfo {author} {\bibfnamefont {D.}~\bibnamefont
  {Mayers}}\ and\ \bibinfo {author} {\bibfnamefont {A.}~\bibnamefont {Yao}},\
  }\href@noop {} {\bibfield  {journal} {\bibinfo  {journal} {Quantum
  Information \& Computation}\ }\textbf {\bibinfo {volume} {4}},\ \bibinfo
  {pages} {273} (\bibinfo {year} {2004})}\BibitemShut {NoStop}%
\bibitem [{\citenamefont {Montanaro}\ and\ \citenamefont {{de
  Wolf}}(2016)}]{Montanaro2016}%
  \BibitemOpen
  \bibfield  {author} {\bibinfo {author} {\bibfnamefont {A.}~\bibnamefont
  {Montanaro}}\ and\ \bibinfo {author} {\bibfnamefont {R.}~\bibnamefont {{de
  Wolf}}},\ }\href {\doibase 10.4086/toc.gs.2016.007} {\bibfield  {journal}
  {\bibinfo  {journal} {Theory of Computing}\ ,\ \bibinfo {pages} {1}}
  (\bibinfo {year} {2016})}\BibitemShut {NoStop}%
\bibitem [{\citenamefont {{\v S}upi{\'c}}\ and\ \citenamefont
  {Bowles}(2019)}]{Supic2019a}%
  \BibitemOpen
  \bibfield  {author} {\bibinfo {author} {\bibfnamefont {I.}~\bibnamefont {{\v
  S}upi{\'c}}}\ and\ \bibinfo {author} {\bibfnamefont {J.}~\bibnamefont
  {Bowles}},\ }\href@noop {} {\  (\bibinfo {year} {2019})},\ \Eprint
  {http://arxiv.org/abs/1904.10042} {arXiv:1904.10042} \BibitemShut {NoStop}%
\bibitem [{\citenamefont {Clauser}\ \emph {et~al.}(1969)\citenamefont
  {Clauser}, \citenamefont {Horne}, \citenamefont {Shimony},\ and\
  \citenamefont {Holt}}]{CHSH}%
  \BibitemOpen
  \bibfield  {author} {\bibinfo {author} {\bibfnamefont {J.~F.}\ \bibnamefont
  {Clauser}}, \bibinfo {author} {\bibfnamefont {M.~A.}\ \bibnamefont {Horne}},
  \bibinfo {author} {\bibfnamefont {A.}~\bibnamefont {Shimony}}, \ and\
  \bibinfo {author} {\bibfnamefont {R.~A.}\ \bibnamefont {Holt}},\ }\href
  {https://link.aps.org/doi/10.1103/PhysRevLett.23.880} {\bibfield  {journal}
  {\bibinfo  {journal} {Phys. Rev. Lett.}\ }\textbf {\bibinfo {volume} {23}},\
  \bibinfo {pages} {880} (\bibinfo {year} {1969})}\BibitemShut {NoStop}%
\bibitem [{\citenamefont {{\v S}upi{\'c}}\ and\ \citenamefont
  {Hoban}(2016)}]{Supic2016a}%
  \BibitemOpen
  \bibfield  {author} {\bibinfo {author} {\bibfnamefont {I.}~\bibnamefont {{\v
  S}upi{\'c}}}\ and\ \bibinfo {author} {\bibfnamefont {M.~J.}\ \bibnamefont
  {Hoban}},\ }\href {\doibase 10.1088/1367-2630/18/7/075006} {\bibfield
  {journal} {\bibinfo  {journal} {New Journal of Physics}\ }\textbf {\bibinfo
  {volume} {18}},\ \bibinfo {pages} {075006} (\bibinfo {year} {2016})},\
  \Eprint {http://arxiv.org/abs/1601.01552} {arXiv:1601.01552} \BibitemShut
  {NoStop}%
\bibitem [{\citenamefont {Gheorghiu}\ \emph {et~al.}(2015)\citenamefont
  {Gheorghiu}, \citenamefont {Kashefi},\ and\ \citenamefont
  {Wallden}}]{Gheorghiu2015}%
  \BibitemOpen
  \bibfield  {author} {\bibinfo {author} {\bibfnamefont {A.}~\bibnamefont
  {Gheorghiu}}, \bibinfo {author} {\bibfnamefont {E.}~\bibnamefont {Kashefi}},
  \ and\ \bibinfo {author} {\bibfnamefont {P.}~\bibnamefont {Wallden}},\ }\href
  {\doibase 10.1088/1367-2630/17/8/083040} {\bibfield  {journal} {\bibinfo
  {journal} {New Journal of Physics}\ }\textbf {\bibinfo {volume} {17}},\
  \bibinfo {pages} {083040} (\bibinfo {year} {2015})}\BibitemShut {NoStop}%
\bibitem [{\citenamefont {Bancal}\ \emph {et~al.}(2018)\citenamefont {Bancal},
  \citenamefont {Sangouard},\ and\ \citenamefont {Sekatski}}]{Bancal2018a}%
  \BibitemOpen
  \bibfield  {author} {\bibinfo {author} {\bibfnamefont {J.-D.}\ \bibnamefont
  {Bancal}}, \bibinfo {author} {\bibfnamefont {N.}~\bibnamefont {Sangouard}}, \
  and\ \bibinfo {author} {\bibfnamefont {P.}~\bibnamefont {Sekatski}},\ }\href
  {\doibase 10.1103/PhysRevLett.121.250506} {\bibfield  {journal} {\bibinfo
  {journal} {Physical Review Letters}\ }\textbf {\bibinfo {volume} {121}},\
  \bibinfo {pages} {250506} (\bibinfo {year} {2018})}\BibitemShut {NoStop}%
\bibitem [{\citenamefont {Renou}\ \emph {et~al.}(2018)\citenamefont {Renou},
  \citenamefont {Kaniewski},\ and\ \citenamefont {Brunner}}]{Renou2018}%
  \BibitemOpen
  \bibfield  {author} {\bibinfo {author} {\bibfnamefont {M.~O.}\ \bibnamefont
  {Renou}}, \bibinfo {author} {\bibfnamefont {J.}~\bibnamefont {Kaniewski}}, \
  and\ \bibinfo {author} {\bibfnamefont {N.}~\bibnamefont {Brunner}},\ }\href
  {\doibase 10.1103/PhysRevLett.121.250507} {\bibfield  {journal} {\bibinfo
  {journal} {Physical Review Letters}\ }\textbf {\bibinfo {volume} {121}},\
  \bibinfo {pages} {250507} (\bibinfo {year} {2018})},\ \Eprint
  {http://arxiv.org/abs/1807.04956} {arXiv:1807.04956} \BibitemShut {NoStop}%
\bibitem [{\citenamefont {Tavakoli}\ \emph {et~al.}(2018)\citenamefont
  {Tavakoli}, \citenamefont {Kaniewski}, \citenamefont {V{\'e}rtesi},
  \citenamefont {Rosset},\ and\ \citenamefont {Brunner}}]{Tavakoli2018c}%
  \BibitemOpen
  \bibfield  {author} {\bibinfo {author} {\bibfnamefont {A.}~\bibnamefont
  {Tavakoli}}, \bibinfo {author} {\bibfnamefont {J.}~\bibnamefont {Kaniewski}},
  \bibinfo {author} {\bibfnamefont {T.}~\bibnamefont {V{\'e}rtesi}}, \bibinfo
  {author} {\bibfnamefont {D.}~\bibnamefont {Rosset}}, \ and\ \bibinfo {author}
  {\bibfnamefont {N.}~\bibnamefont {Brunner}},\ }\href {\doibase
  10.1103/PhysRevA.98.062307} {\bibfield  {journal} {\bibinfo  {journal}
  {Physical Review A}\ }\textbf {\bibinfo {volume} {98}},\ \bibinfo {pages}
  {062307} (\bibinfo {year} {2018})}\BibitemShut {NoStop}%
\bibitem [{\citenamefont {Sekatski}\ \emph {et~al.}(2018)\citenamefont
  {Sekatski}, \citenamefont {Bancal}, \citenamefont {Wagner},\ and\
  \citenamefont {Sangouard}}]{Sekatski2018a}%
  \BibitemOpen
  \bibfield  {author} {\bibinfo {author} {\bibfnamefont {P.}~\bibnamefont
  {Sekatski}}, \bibinfo {author} {\bibfnamefont {J.-D.}\ \bibnamefont
  {Bancal}}, \bibinfo {author} {\bibfnamefont {S.}~\bibnamefont {Wagner}}, \
  and\ \bibinfo {author} {\bibfnamefont {N.}~\bibnamefont {Sangouard}},\ }\href
  {\doibase 10.1103/PhysRevLett.121.180505} {\bibfield  {journal} {\bibinfo
  {journal} {Physical Review Letters}\ }\textbf {\bibinfo {volume} {121}},\
  \bibinfo {pages} {180505} (\bibinfo {year} {2018})}\BibitemShut {NoStop}%
\bibitem [{\citenamefont {Sperling}\ and\ \citenamefont
  {Vogel}(2011)}]{Sperling2011}%
  \BibitemOpen
  \bibfield  {author} {\bibinfo {author} {\bibfnamefont {J.}~\bibnamefont
  {Sperling}}\ and\ \bibinfo {author} {\bibfnamefont {W.}~\bibnamefont
  {Vogel}},\ }\href {\doibase 10.1088/0031-8949/83/04/045002} {\bibfield
  {journal} {\bibinfo  {journal} {Physica Scripta}\ }\textbf {\bibinfo {volume}
  {83}},\ \bibinfo {pages} {045002} (\bibinfo {year} {2011})}\BibitemShut
  {NoStop}%
\bibitem [{\citenamefont {Buscemi}(2012{\natexlab{b}})}]{Buscemi2012LOSR}%
  \BibitemOpen
  \bibfield  {author} {\bibinfo {author} {\bibfnamefont {F.}~\bibnamefont
  {Buscemi}},\ }\href {https://link.aps.org/doi/10.1103/PhysRevLett.108.200401}
  {\bibfield  {journal} {\bibinfo  {journal} {Phys. Rev. Lett.}\ }\textbf
  {\bibinfo {volume} {108}},\ \bibinfo {pages} {200401} (\bibinfo {year}
  {2012}{\natexlab{b}})}\BibitemShut {NoStop}%
\bibitem [{\citenamefont {Buscemi}(2016)}]{Buscemi-ProbInfTrans-2016}%
  \BibitemOpen
  \bibfield  {author} {\bibinfo {author} {\bibfnamefont {F.}~\bibnamefont
  {Buscemi}},\ }\href {\doibase 10.1134/S0032946016030017} {\bibfield
  {journal} {\bibinfo  {journal} {Probl Inf Transm}\ }\textbf {\bibinfo
  {volume} {52}},\ \bibinfo {pages} {201} (\bibinfo {year} {2016})}\BibitemShut
  {NoStop}%
\bibitem [{\citenamefont {Rosset}\ \emph
  {et~al.}(2018{\natexlab{b}})\citenamefont {Rosset}, \citenamefont {Buscemi},\
  and\ \citenamefont {Liang}}]{Rosset-Buscemi-Liang-PRX}%
  \BibitemOpen
  \bibfield  {author} {\bibinfo {author} {\bibfnamefont {D.}~\bibnamefont
  {Rosset}}, \bibinfo {author} {\bibfnamefont {F.}~\bibnamefont {Buscemi}}, \
  and\ \bibinfo {author} {\bibfnamefont {Y.-C.}\ \bibnamefont {Liang}},\ }\href
  {\doibase 10.1103/PhysRevX.8.021033} {\bibfield  {journal} {\bibinfo
  {journal} {Phys. Rev. X}\ }\textbf {\bibinfo {volume} {8}},\ \bibinfo {pages}
  {021033} (\bibinfo {year} {2018}{\natexlab{b}})}\BibitemShut {NoStop}%
\bibitem [{\citenamefont {Skrzypczyk}\ and\ \citenamefont
  {Linden}(2019)}]{Skr-Linden-2019}%
  \BibitemOpen
  \bibfield  {author} {\bibinfo {author} {\bibfnamefont {P.}~\bibnamefont
  {Skrzypczyk}}\ and\ \bibinfo {author} {\bibfnamefont {N.}~\bibnamefont
  {Linden}},\ }\href {\doibase 10.1103/PhysRevLett.122.140403} {\bibfield
  {journal} {\bibinfo  {journal} {Phys. Rev. Lett.}\ }\textbf {\bibinfo
  {volume} {122}},\ \bibinfo {pages} {140403} (\bibinfo {year}
  {2019})}\BibitemShut {NoStop}%
\bibitem [{\citenamefont {Skrzypczyk}\ \emph {et~al.}(2019)\citenamefont
  {Skrzypczyk}, \citenamefont {Supic},\ and\ \citenamefont
  {Cavalcanti}}]{Skr-Supic-Cavalcanti-2019}%
  \BibitemOpen
  \bibfield  {author} {\bibinfo {author} {\bibfnamefont {P.}~\bibnamefont
  {Skrzypczyk}}, \bibinfo {author} {\bibfnamefont {I.}~\bibnamefont {Supic}}, \
  and\ \bibinfo {author} {\bibfnamefont {D.}~\bibnamefont {Cavalcanti}},\
  }\href {\doibase 10.1103/PhysRevLett.122.130403} {\bibfield  {journal}
  {\bibinfo  {journal} {Phys. Rev. Lett.}\ }\textbf {\bibinfo {volume} {122}},\
  \bibinfo {pages} {130403} (\bibinfo {year} {2019})}\BibitemShut {NoStop}%
\bibitem [{\citenamefont {Takagi}\ \emph {et~al.}(2019)\citenamefont {Takagi},
  \citenamefont {Regula}, \citenamefont {Bu}, \citenamefont {Liu},\ and\
  \citenamefont {Adesso}}]{Takagi-etal-2019}%
  \BibitemOpen
  \bibfield  {author} {\bibinfo {author} {\bibfnamefont {R.}~\bibnamefont
  {Takagi}}, \bibinfo {author} {\bibfnamefont {B.}~\bibnamefont {Regula}},
  \bibinfo {author} {\bibfnamefont {K.}~\bibnamefont {Bu}}, \bibinfo {author}
  {\bibfnamefont {Z.-W.}\ \bibnamefont {Liu}}, \ and\ \bibinfo {author}
  {\bibfnamefont {G.}~\bibnamefont {Adesso}},\ }\href {\doibase
  10.1103/PhysRevLett.122.140402} {\bibfield  {journal} {\bibinfo  {journal}
  {Phys. Rev. Lett.}\ }\textbf {\bibinfo {volume} {122}},\ \bibinfo {pages}
  {140402} (\bibinfo {year} {2019})}\BibitemShut {NoStop}%
\bibitem [{\citenamefont {Uola}\ \emph
  {et~al.}(2019{\natexlab{b}})\citenamefont {Uola}, \citenamefont {Kraft},
  \citenamefont {Shang}, \citenamefont {Yu},\ and\ \citenamefont
  {G\"uhne}}]{Uola-etal-2019}%
  \BibitemOpen
  \bibfield  {author} {\bibinfo {author} {\bibfnamefont {R.}~\bibnamefont
  {Uola}}, \bibinfo {author} {\bibfnamefont {T.}~\bibnamefont {Kraft}},
  \bibinfo {author} {\bibfnamefont {J.}~\bibnamefont {Shang}}, \bibinfo
  {author} {\bibfnamefont {X.-D.}\ \bibnamefont {Yu}}, \ and\ \bibinfo {author}
  {\bibfnamefont {O.}~\bibnamefont {G\"uhne}},\ }\href {\doibase
  10.1103/PhysRevLett.122.130404} {\bibfield  {journal} {\bibinfo  {journal}
  {Phys. Rev. Lett.}\ }\textbf {\bibinfo {volume} {122}},\ \bibinfo {pages}
  {130404} (\bibinfo {year} {2019}{\natexlab{b}})}\BibitemShut {NoStop}%
\bibitem [{\citenamefont {Uola}\ \emph
  {et~al.}(2019{\natexlab{c}})\citenamefont {Uola}, \citenamefont {Kraft},\
  and\ \citenamefont {Abbott}}]{uola2019quantification}%
  \BibitemOpen
  \bibfield  {author} {\bibinfo {author} {\bibfnamefont {R.}~\bibnamefont
  {Uola}}, \bibinfo {author} {\bibfnamefont {T.}~\bibnamefont {Kraft}}, \ and\
  \bibinfo {author} {\bibfnamefont {A.~A.}\ \bibnamefont {Abbott}},\
  }\href@noop {} {\enquote {\bibinfo {title} {Quantification of quantum
  dynamics with input-output games},}\ } (\bibinfo {year}
  {2019}{\natexlab{c}}),\ \Eprint {http://arxiv.org/abs/1906.09206}
  {arXiv:1906.09206 [quant-ph]} \BibitemShut {NoStop}%
\bibitem [{\citenamefont {Takagi}\ and\ \citenamefont
  {Regula}(2019)}]{Takagi-Regula-PRX}%
  \BibitemOpen
  \bibfield  {author} {\bibinfo {author} {\bibfnamefont {R.}~\bibnamefont
  {Takagi}}\ and\ \bibinfo {author} {\bibfnamefont {B.}~\bibnamefont
  {Regula}},\ }\href {\doibase 10.1103/PhysRevX.9.031053} {\bibfield  {journal}
  {\bibinfo  {journal} {Phys. Rev. X}\ }\textbf {\bibinfo {volume} {9}},\
  \bibinfo {pages} {031053} (\bibinfo {year} {2019})}\BibitemShut {NoStop}%
\end{thebibliography}%

\onecolumngrid
\bigskip

\begin{center}
\line(1,0){250}
\end{center}

\bigskip

\twocolumngrid

\onecolumngrid
\clearpage


\end{document}